\documentclass[a4paper,UKenglish]{lipics-v2019}

\hideLIPIcs

\usepackage{multicol}
\usepackage{amsmath}
\usepackage{amsthm}
\usepackage{amssymb}
\usepackage{parsetree}
\usepackage{url}
\usepackage{xspace}

\newcommand{\ignore}[1]{}

\nolinenumbers 

\def\calT{{\mathcal T}}
\def\calW{{\mathcal W}}
\def\calI{{\mathcal I}}

\def\calS{{\mathcal S}}
\def\calR{{\mathcal R}}
\def\calRS{{\calR}_{\calS}}
\def\TtoW{\calT2\calW}
\def\WtoT{\calW2\calT}

\newcommand{\mc}[1]{\mathcal{ #1}}
\newcommand{\boxtheorem}{\hfill $\Box$}
\newcommand{\nit}[1]{{\it #1}}

\newcommand{\mf}[1]{\mathfrak{ #1}}

\newcommand{\an}[1]{\!\!:\!\!{#1}}

\newcommand{\dpm}{Datalog$^{\pm}$\xspace}
\newcommand{\dpmsn}{Datalog$^{\pm,\neg sg}$\xspace}

\newcommand{\mult}[2]{\nit{mult}(#1,#2)}

\newcommand{\pt}{\textsf{ProofTree}\xspace}
\newcommand{\head}{{\it head}}
\newcommand{\body}{{\it body}}
\newcommand{\depth}{{\it depth}}
\newcommand{\chase}{{\it chase}}

\newcommand{\var}{{\it var}}
\newcommand{\nulls}{{\it nulls}}
\newcommand{\length}{{\it length}}
\newcommand{\ptt}{\textsf{Multiplicity}\xspace}
\newcommand{\Tab}{{\it Tab}}
\newcommand{\Stack}{{\it Stack}}

\newcommand{\tgd}{{\em tgd}}

\newcommand{\bl}[1]{{#1}}
\newcommand{\green}[1]{\textcolor{green}{#1}}
\newcommand{\red}[1]{\textcolor{red}{#1}}
\newcommand{\blue}[1]{\textcolor{blue}{#1}}

\title{ Datalog: Bag Semantics via Set Semantics}
\titlerunning{Datalog: Bag Semantics via Set Semantics}

\author{Leopoldo Bertossi}{RelationalAI \ Inc., USA \ and \ Carleton University,  Ottawa, Canada\\
Member of the ``Millenium Institute for Foundational Research on Data'' (IMFD, Chile)}
{bertossi@scs.carleton.ca}{}{}

\author{Georg Gottlob}{University of Oxford, UK and
TU Wien, Austria}
{georg.gottlob@cs.ox.ac.uk}{}{}

\author{Reinhard Pichler}
{TU Wien, Austria}{pichler@dbai.tuwien.ac.at}{}{}

\authorrunning{L.~Bertossi, G.~Gottlob, and R.~Pichler}

\Copyright{Leopoldo Bertossi, Georg Gottlob and Reinhard Pichler}

\ccsdesc{Information systems -- Data management systems -- Query languages, 
Theory of computation -- Logic,  Theory of computation -- Semantics and reasoning}

\keywords{Datalog, duplicates, multisets, query answering, chase, Datalog$^\pm$}

\relatedversion{A full version of the paper is available at \url{http://arxiv.org/abs/1803.06445}.}

\acknowledgements{Many thanks to Renzo Angles and Claudio Gutierrez for information on their work on SPARQL  with bag semantics; and to Wolfgang Fischl for his help testing some queries in SQL DBMSs.
We appreciate the useful comments received from the reviewers.
Part of this work was done while L.~Bertossi was spending a sabbatical at the
DBAI Group of TU Wien with support
from the ``Vienna Center for Logic and Algorithms" and the Wolfgang Pauli Society.
This author is grateful for their support and hospitality, and specially to G.~Gottlob for making the stay possible. He was also supported by NSERC Discovery Grant \#06148.
The work of G.~Gottlob was supported by the Austrian Science Fund (FWF):P30930 and the EPSRC programme grant EP/M025268/1 VADA.
The work of R.~Pichler was supported by the Austrian Science Fund (FWF):P30930.}

\EventEditors{Pablo Barcelo and Marco Calautti}
\EventNoEds{2}
\EventLongTitle{22nd International Conference on Database Theory (ICDT 2019)}
\EventShortTitle{ICDT 2019}
\EventAcronym{ICDT}
\EventYear{2019}
\EventDate{March 26--28, 2019}
\EventLocation{Lisbon, Portugal}
\EventLogo{}
\SeriesVolume{127}
\ArticleNo{13}

\begin{document}
\maketitle

\begin{abstract}
Duplicates in data management are common and problematic.
In this work, we present a translation of Datalog under bag semantics into a
well-behaved extension of Datalog, the so-called {\em warded Datalog}$^\pm$, under set semantics.
From a theoretical point of view, this allows us to reason on bag semantics
by making use of the well-established theoretical foundations of set semantics.
From a practical point of view, this allows us to handle the bag semantics of
Datalog by powerful, existing query engines for the required extension of Datalog. This use of Datalog$^\pm$ is extended to give a set semantics to duplicates in Datalog$^\pm$ itself.
We investigate the properties of the resulting Datalog$^\pm$ programs, the problem of deciding multiplicities,
and expressibility of some bag operations.
Moreover, the proposed translation has the potential for interesting applications such as
to Multiset Relational Algebra
and the semantic web query language SPARQL with bag semantics.
\end{abstract}


\section{Introduction}
\label{sect:introduction}

Duplicates are a common feature in data management. They appear, for instance, in the result of SQL queries over relational databa\-ses or when a
SPARQL query is posed over RDF data.
However, the semantics of data operations and queries in the presence of duplicates is not always clear, because duplicates are handled by bags or multisets, but common logic-based semantics in data management are set-theoretical, making it difficult to tell apart duplicates through the use of sets alone.
To address this problem, a bag semantics for Datalog programs was proposed in \cite{DBLP:conf/vldb/MumickPR90}, what we refer to as  the {\em derivation-tree bag semantics} (DTB semantics). Intuitively, two duplicates of the same tuple in an intentional predicate are accepted as such  if they have syntactically different derivation trees.  The DTB semantics was used in
\cite{DBLP:conf/semweb/AnglesG16} to provide a bag semantics for SPARQL.


\ignore{*************
The DTB semantics has two important drawbacks: first, it is operational --
thus losing the declarative, logic-based semantics of Datalog; and second, it is defined via new constructs (the derivation-trees) --
thus leaving the world of established query language concepts and losing the
applicability of the large body of query optimization techniques.
*************}

The DTB semantics follows a proof-theoretic approach, which requires external, meta-level reasoning over the set of all derivation trees
rather than allowing for a query\ignore{/programming} language that inherently collects those duplicates.
The main goal of this paper is to identify a syntactic class of extended
Datalog programs, such that: (a) it extends classical Datalog with stratified negation and has a classical model-based semantics,
\ignore{*************
(b)  it can be used for representing duplicates and computing with them
according to a natural and intended bag semantics,
*************}
(b) for every program in the class with a bag semantics\ignore{ for duplicates},
another program in the same class can be built that has a set-semantics and fully captures the bag semantics of the initial program, (c) it can be used in particular to give a set-semantics for classical Datalog with stratified negation with bag semantics\ignore{ for duplicates}. All these results can be applied, in particular,  to multi-relational algebra, i.e. relational algebra with duplicates.

To this end, we show that the DTB semantics of a Datalog program can be represented
by means of its transformation into a Datalog$^\pm$  program
\cite{DBLP:journals/jair/CaliGK13,DBLP:journals/ws/CaliGL12},
in such a way that the intended model of the former, including duplicates, can be characterized as the result of the duplicate-free chase instance for the latter.
The crucial idea of our translation from bag semantics into set semantics (of \dpm) is the
introduction of tuple ids (tids) via existentially quantified variables in the rule heads.
Different tids of the same tuple will allow us to identify usual duplicates when falling back to a bag semantics for the original Datalog program. We establish the correspondence between the DTB semantics and ours. This correspondence is then extended to Datalog with stratified negation.
We thus%
\ignore{************
achieve a fully declarative
way of expressing the bag semantics of an important query language; and we immediately
************}
recover full relational algebra (including set difference) with bag semantics in terms of a well-behaved query language under set semantics.

The programs we use for this task belong to {\em warded} Datalog$^\pm$
\cite{DBLP:conf/ijcai/GottlobP15}. This is a particularly well-behaved class of programs in that it
properly extends Datalog, has a tractable
conjunctive query answering (CQA) problem, and
has recently been implemented in a powerful query engine, namely the VADALOG System
\cite{DBLP:conf/ijcai/BellomariniGPS17,pvldb/BellomariniSG18}.
None of the other well-known
classes of Datalog$^\pm$ share these properties:
for instance, guarded \cite{DBLP:journals/jair/CaliGK13}, sticky  and weakly-sticky \cite{DBLP:journals/ai/CaliGP12} \dpm only allow restricted forms of joins and, hence, do not cover Datalog.
On the other hand, more expressive languages, such as weakly frontier guarded  \dpm \cite{DBLP:journals/ai/BagetLMS11},
lose tractability of CQA.
Warded Datalog$^\pm$
has been successfully applied to represent a core fragment of SPARQL under
certain OWL 2 QL entailment regimes \cite{glimm}, with set semantics though \cite{DBLP:conf/ijcai/GottlobP15}
(see also \cite{DBLP:conf/pods/ArenasGP14,DBLP:conf/pods/ArenasGP18}),
and it looks promising as a general language for specifying different data management tasks \cite{DBLP:conf/ijcai/BellomariniGPS17}.

\ignore{**********
We thus eliminate the need to resort to an operational semantics to be handled outside the query language.
**********}

We then go one step further and also express the bag semantics of \dpm
by means of the set semantics of \dpm. In fact, we show that the bag semantics of
a very general language in the \dpm class
can be expressed via the set semantics of Datalog$^\pm$  and the transformed program is
warded whenever the initial program is.

\smallskip
\noindent
{\bf Structure and main results.}
In Section~\ref{sec:preliminaries}, we recall some basic notions.
In Section~\ref{sec:conclusion}, we conclude and
discuss some future extensions.
The main results of the paper are detailed in Sections \ref{sec:datalog} -- \ref{sec:next}.
\begin{itemize}
\item Our translation of Datalog with bag semantics into
warded Datalog$^{\pm}$ with set semantics,
which will be referred to as program-based bag (PBB ) semantics,
is presented in Section \ref{sec:datalog}. We also show how this translation can be extended
to Datalog with stratified negation.\ignore{ and how these results allow us to capture the
multiset relational algebra (MRA).}
\item  In Section~\ref{sec:warded}, we study the transformation from bag semantics into set
semantics for \dpm itself. We thus carry over both the DTB semantics and the PBB semantics
to \dpm with a form of stratified negation, and establish the equivalence of these two semantics also for this
extended query language. Moreover, we \ignore{will} verify that the \dpm programs resulting from our
transformation are warded whenever the programs to start with belong to this class.
\item In Section~\ref{sec:decidability}, we study crucial decision problems related to multiplicities. Above all, we are interested in the question if a given tuple has finite or infinite  multiplicity. Moreover, in case of finiteness, we want to compute the precise multiplicity. We show that
these tasks can be accomplished in polynomial time (data complexity).
\item In Section \ref{sec:next}, we apply our results on Datalog with bag semantics to
Multiset Relational Algebra (MRA). We also discuss
alternative semantics for multiset-intersection and multiset-difference,
and the difficulties to capture them with our \dpm approach.
\end{itemize}


\section{Preliminaries}
\label{sec:preliminaries}

We assume familiarity with the relational data model,  conjunctive queries (CQs), in particular Boolean conjunctive queries (BCQs);  classical Datalog with minimal-model semantics, and  Datalog with  stratified negation with standard-model semantics, denoted Datalog$^{\neg s}$
(see  \cite{AHV94} for an  introduction).
\ignore{We introduce some useful notation here:
for a relational predicate $P$ and an instance $D$, $P(D)$ denotes the extension of $P$ in $D$. } An $n$-ary relational predicate $P$ has {\em positions}: $P[1], \ldots, P[n]$. With $\nit{Pos}(P)$ we denote the set of positions of predicate $P$; and with $\nit{Pos}(\Pi)$ the set of positions of (predicates in) a program $\Pi$.

\subsection{Derivation-Tree Bag (DTB) Semantics for Datalog and Datalog$^{\neg s}$}
\label{sec:dtbs}

We follow \cite{DBLP:conf/vldb/MumickPR90}, where tuples are {\em colored\/}
to tell apart duplicates of a same element in the extensional database (EDB), via an infinite, ordered list $\mc{C}$ of colors \ $c_1, c_2, \ldots$.
For  a  multiset $M$, \ $e \in M$ if $\mult{e}{M} >0$  holds,
where $\mult{e}{M}$ denotes the multiplicity of $e$ in $M$.
In this case, the $n$ copies of $e$ are denoted by $e\an{1}, \ldots, e\an{n}$, indicating that they are colored with $c_1, \ldots, c_n$,
respectively.  So,
$\nit{col}(e) := \{e\an{1}, \ldots, e\an{n}\}$ becomes a set. A multiset $M_1$ is (multi)contained in multiset $M_2$, when $\mult{e}{M_1} \leq \mult{e}{M_2}$ for every $e \in M_1$. For a multiset $M$, $\nit{col}(M) := \bigcup_{e \in M} \nit{col}(e)$, which is a set. For a ``colored" set $S$, $\nit{col}^{-1}(S)$ produces a multiset by stripping tuples from their colors.

\begin{example} \label{ex:multi}
For $M = \{a,a,a, b,b, c\}$,  $\nit{col}(a) = \{a\an{1},a\an{2},a\an{3}\}$, and
$\nit{col}(M) = S = \{a\an{1},a\an{2},a\an{3},b\an{1}, b\an{2},c\an{1}\}$.
The inverse operation, the decoloration, gives, for instance:
$\nit{col}^{-1}(a\an{2}) := a$; and  $\nit{col}^{-1}(S) := \{a,a,a, b,b, c\}$,  a multiset.
\boxtheorem
\end{example}

We consider  Datalog programs $\Pi$ with multiset predicates and multiset EDBs  $D$. A {\em derivation tree} (DT) for $\Pi$ wrt.  $D$
is a tree with labeled nodes and edges, as follows:
\begin{enumerate}
\item For an EDB predicate $P$ and $h \in \nit{col}(P(D))$,  a DT for $\nit{col}^{-1}(h)$ contains a single node with label  $h$.

\item For each rule of the form \quad
$\rho \!: \ H \leftarrow A_1,A_2,\ldots,A_k$; \ \ with $k>0$, \
and each tuple $\langle \calT_1,\ldots, \calT_k\rangle$ of DTs for the atoms $\langle d_1, \ldots, d_k\rangle$ that unify with  $\langle A_1,A_2,\ldots,A_k\rangle$ with mgu $\theta$, generate a DT for $H\theta$ with $H\theta$ as the root label, $\langle \calT_1,\ldots, \calT_k\rangle$ as the children, and $\rho$ as the label for the edges from the root to the children.
We assume that these children are arranged in the order of the corresponding body
atoms in rule $\rho$.
\end{enumerate}

For a DT $\calT$, we define $\nit{Atoms}(\calT)$ as $\nit{col}^{-1}$ of the root when $\calT$ is a single-node tree, and the root-label of  $\calT$, otherwise. \ignore{ := \mbox{ root-label of } \calT \mbox{ or } $\nit{col}^{-1}$ of the root, with the latter case  when $\calT$ is a single-node tree.}
For a set of DTs $\mf{T}$: \ $\nit{Atoms}(\mf{T}) := \biguplus \{\nit{Atoms}(\calT)~|~\calT\in \mf{T}\}$, which is a multiset that multi-contains $D$. Here, we write $\biguplus$ to denote the multi-union of (multi)sets that keeps all duplicates.
If $\nit{DT}(\Pi,$ $\nit{col}(D))$ is the set of (syntactically different) DTs\ignore{ for $\Pi$ wrt. $D$},  the {\em derivation-tree bag (DTB) semantics} for $\Pi$ is the multiset \ \
$\nit{DTBS}(\Pi,D) := \nit{Atoms}(\nit{DT}(\Pi,\nit{col}(D)))$. (Cf. \cite{DBLP:conf/pods/GreenKT07} for an equivalent, provenance-based bag semantics for Datalog.)


\begin{figure}[t]
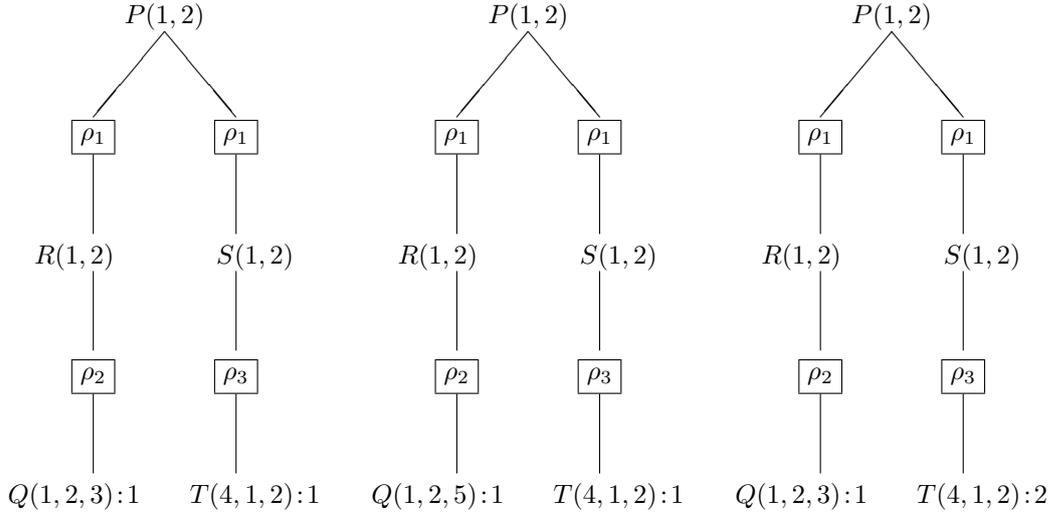

\begin{multicols}{3}
{\tiny
\hspace*{-1.5cm}\ptbegtree
\ptbeg
\ptnode{\centerline{$P(1,2)$}}
\ptbeg
\ptnode{\fbox{$\rho_1$}} \ptbeg \ptnode{$R(1,2)$\hspace{5mm}}\ptbeg \ptnode{\fbox{$\rho_2$}}  \ptleaf{$Q(1,2,3)\!:\!1$\hspace{5mm}} \ptend \ptend \ptend
\ptbeg
\ptnode{\fbox{$\rho_1$}} \ptbeg \ptnode{\hspace{5mm}$S(1,2)$} \ptbeg \ptnode{\fbox{$\rho_3$}}  \ptleaf{\hspace{5mm}$T(4,1,2)\!:\!1$} \ptend \ptend \ptend
\ptend
\ptendtree
}

{\tiny
\hspace*{-1.5cm}\ptbegtree
\ptbeg
\ptnode{\centerline{$P(1,2)$}}
\ptbeg
\ptnode{\fbox{$\rho_1$}} \ptbeg \ptnode{$R(1,2)$\hspace{5mm}}\ptbeg \ptnode{\fbox{$\rho_2$}}  \ptleaf{$Q(1,2,5)\!:\!1$\hspace{5mm}} \ptend \ptend \ptend
\ptbeg
\ptnode{\fbox{$\rho_1$}} \ptbeg \ptnode{\hspace{5mm}$S(1,2)$} \ptbeg \ptnode{\fbox{$\rho_3$}}  \ptleaf{\hspace{5mm}$T(4,1,2)\!:\!1$} \ptend \ptend \ptend
\ptend
\ptendtree
}

{\tiny
\hspace*{-1.5cm}\ptbegtree
\ptbeg
\ptnode{\centerline{$P(1,2)$}}
\ptbeg
\ptnode{\fbox{$\rho_1$}} \ptbeg \ptnode{$R(1,2)$\hspace{5mm}}\ptbeg \ptnode{\fbox{$\rho_2$}}  \ptleaf{$Q(1,2,3)\!:\!1$\hspace{5mm}} \ptend \ptend \ptend
\ptbeg
\ptnode{\fbox{$\rho_1$}} \ptbeg \ptnode{\hspace{5mm}$S(1,2)$} \ptbeg \ptnode{\fbox{$\rho_3$}}  \ptleaf{\hspace{5mm}$T(4,1,2)\!:\!2$} \ptend \ptend \ptend
\ptend
\ptendtree
}
\end{multicols}
\caption{Three (out of four) derivation-trees for $P(1,2)$ in Example~\ref{ex:theex}.}
\label{fig:derivation-trees-example-two}
\end{figure}

\begin{example} \label{ex:theex}
Consider the program $\Pi = \{\rho_1, \rho_2, \rho_3\}$ and multiset EDB $D=\{Q(1,2,3)$,
\linebreak $Q(1,2,5), Q(2,3,4), Q(2,3,4), T(4,1,2), T(4,1,2) \}$,
where $\rho_1, \rho_2, \rho_3$ are defined as follows:
\\[1.1ex]
\hspace*{5mm}$\rho_1\!: \ P(x,y)  \leftarrow R(x,y), S(x,y)$; \quad
$\rho_2\!: \ R(x,y) \leftarrow Q(x,y,z)$; \quad
$\rho_3\!: \ S(x,y) \leftarrow T(z,x,y).$

\smallskip

\noindent
Here, $\nit{col}(D) = \{Q(1,2,3)\an{1}, Q(1,2,5)\an{1}, Q(2,3,4)\an{1}, Q(2,3,4)\an{2},T(4,1,2)\an{1},T(4,1,2)\an{2} \}$.
In total, we have 16 DTs: \ (a)
6 single-node trees with labels in $\nit{col}(D)$  (b)
6 depth-two, linear trees (root to the left, i.e. rotated in $- 90^\circ$):
$R(1,2) - $ $\rho_2 - Q(1,2,3)\an{1}$.
$R(1,2) - \rho_2 - {Q(1,2,5)\an{1}}$.
$R(2,3) - \rho_2 - Q(2,3,4)\an{1}$.
$R(2,3) - \rho_2 - Q(2,3,4)\an{2}$.
$S(1,2) - \rho_3 - T(4,1,2)\an{1}$.
$S(1,2) - \rho_3 - T(4,1,2)\an{2}$.
(c) 4~depth-three trees for $P(1,2)$, three of which are
displayed in Figure \ref{fig:derivation-trees-example-two}.
The 16 different DTs in $\nit{DT}(\Pi,\nit{col}(D))$ give rise to
$\nit{DTBS}(\Pi,D) = D \ \cup \ \{R(1,2),$ $ R(1,2),$ $R(2,3)$, $R(2,3),$ $S(1,2)$,
$S(1,2),$  $ P(1,2),$  $P(1,2),$ $P(1,2),$ $ P(1,2)\}$.
 \boxtheorem
\end{example}

\ignore{*************
Notice that the top rule in Example \ref{ex:theex} contains an implicit intersection. With the DTB semantics it behaves as a join, and we obtain the multiplication of multiplicities for it.
*************}

In \cite{DBLP:conf/adc/MumickS93}, a bag semantics for Datalog$^{\neg s}$ was introduced
via derivation-trees (DTs), extending the DTB semantics in \cite{DBLP:conf/vldb/MumickPR90}
for (positive) Datalog.
This extension applies to Datalog programs with stratified negation that are range-restricted and safe, i.e. a variable in
a rule head  or in a negative literal must appear in a positive literal in the body of the same rule.
(The semantics in \cite{DBLP:conf/adc/MumickS93} does not consider duplicates in the EDB, but
it is easy to extend the DTB semantics  with multiset EDBs using colorings as above.) \ If $\Pi$
is a Datalog$^{\neg s}$ program with multiset predicates and multiset EDB  $D$, a {\em derivation
tree}  for $\Pi$ wrt.  $D$ is as for Datalog programs, but with condition {\bf \sf{2.}} modified as follows:
%
\begin{itemize}

\item[{\bf \sf{2'.}}\hspace{-1mm}] Now let $\rho$ be a rule of the form
$\rho\!: \ H \leftarrow A_1,A_2,\ldots,A_k, \neg B_1, \dots \neg B_\ell$; \ \ with $k>0$ and $\ell \geq 0$.
Let the predicate of $H$ be of some stratum $i$ and let the predicates of $B_1, \dots, B_\ell$ be of some stratum $< i$.
Assume that we have already computed all derivation trees for (atoms with) predicates up to stratum $i-1$.
Then, for each tuple $\langle \calT_1,\ldots, \calT_k\rangle$ of DTs for the atoms $\langle d_1, \ldots, d_k\rangle$ that unify with  $\langle A_1,A_2,\ldots,A_k\rangle$ with mgu $\theta$, such that there is no DT for any of the atoms
$B_1 \theta, \dots, B_\ell \theta$,
generate a DT for $H\theta$ with $H\theta$ as the root label and $\langle \calT_1,\ldots, \calT_k,
\neg B_1 \theta, \dots, \neg B_\ell \theta \rangle$ as the children, in this order.
Furthermore, all edges from the root to its children are labelled with $\rho$.
\end{itemize}

\noindent
Analogously to the positive case, now for a range-restricted and safe program $\Pi$ in Datalog$^{\neg s}$ and multiset EDB $D$, we write $\nit{DTBS}(\Pi,D)$
to denote the derivation-tree based bag semantics.

\begin{example} \ \label{ex:diff} (ex.\ \ref{ex:theex} cont.) Consider now the EDB $D' = \{Q(1,2,3)$,
 $Q(1,2,5), Q(2,3,4),$  $Q(2,3,4), T(4,1,2)\}$ (with one duplicate of $T(4,1,2)$ removed from $D$), and modify $\Pi$ to $\Pi' = \{\rho'_1, \rho_2,\rho_3\}$ with
$\rho'_1\!: \ \ P(x,y) \leftarrow R(x,y), \neg S(x,y)$ \quad (i.e., $\rho'_1$ now encodes multiset difference). \
Then, predicates $Q,R,S,T$ are on stratum 0 and $P$ is on stratum 1. The DTs for atoms with predicates from stratum 0
are as in  Example \ref{ex:theex}
with two exceptions:
there is now only one single-node DT for $T(4,1,2)$ and only one DT for $S(1,2)$.

For ground atoms with predicate $P$, we now only get two DTs producing $P(2,3)$.
One of them is shown in Figure~\ref{fig:derivation-trees-example-three}
on the left-hand side.
The other DT of $P(2,3)$ is obtained by replacing the left-most leave $Q(2,3,4)\!\!:\!\!1$ by $Q(2,3,4)\!\!:\!\!2$.
In particular, the DT of $P(2,3)$ in Figure~\ref{fig:derivation-trees-example-three}
shows that the ``derivation'' of  $\neg S(2,3)$ succeeds, i.e., there is no DT for $S(2,3)$.
The remaining two trees in Figure~\ref{fig:derivation-trees-example-three}
establish that $P(2,3)$ do not have a DT because $S(1,2)$ does have a DT.
In total, we have 12 different DTs in $\nit{DT}(\Pi',\nit{col}(D'))$ with
$\nit{DTBS}(\Pi',D') = D' \ \cup \ \{R(1,2),$ $ R(1,2),$ $R(2,3),$ $R(2,3),$
$S(1,2),$  $ P(2,3),$  $P(2,3)\}$.

\ignore{\hspace*{1cm}$P(2,3) - \rho_1- \frac{\neg S(2,3)~~ \times}
{R(2,3) - \rho_2 - Q(2,3,4) : 1}$
\quad \quad \quad
$P(2,3) - \rho_1- \frac{\neg S(2,3) ~~ \times}
{R(2,3) - \rho_2 - Q(2,3,4) : 2}$,   }

Notice that the DT semantics interprets the difference in rule $\rho_1'$ as ``all or nothing": when computing $P(x,y)$, a single DT for $S(x,y)$ ``kills" all the DTs for $R(x,y)$ \ (cf. Section \ref{sec:next}). For example, $P(1,2)$ is not obtained despite the fact that we have two copies of $R(1,2)$ and only one of $S(1,2)$, as the two trees on the right-hand side in Figure~\ref{fig:derivation-trees-example-three}
show.\boxtheorem
\end{example}

\begin{figure}[t]
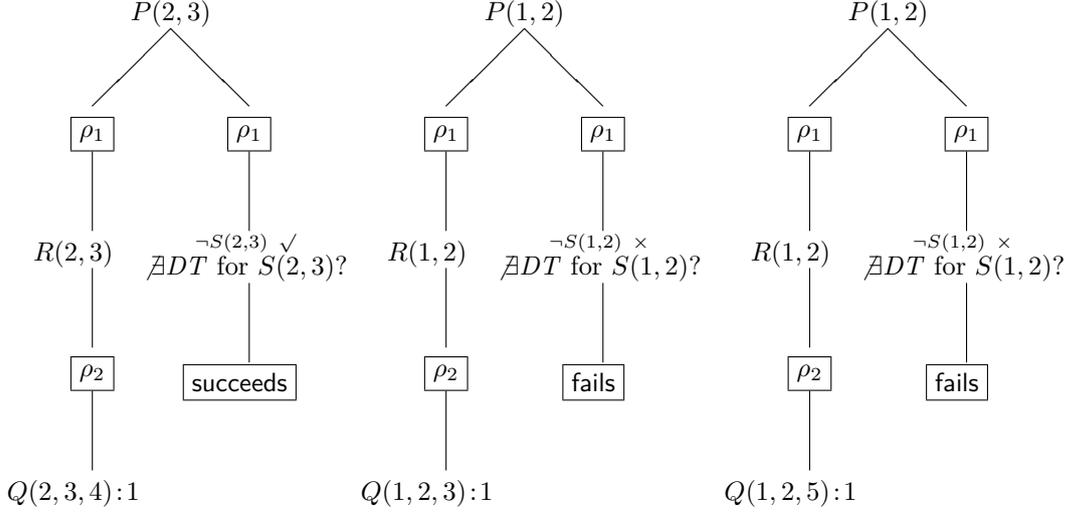

\begin{multicols}{3}
{\tiny
\hspace*{-0.5cm}\ptbegtree
\ptbeg
\ptnode{\centerline{$P(2,3)$}}
\ptbeg
\ptnode{\fbox{$\rho_1$}} \ptbeg \ptnode{$R(2,3)$\hspace{5mm}}\ptbeg \ptnode{\fbox{$\rho_2$}}  \ptleaf{$Q(2,3,4)\!:\!1$\hspace{5mm}} \ptend \ptend \ptend
\ptbeg
\ptnode{\fbox{$\rho_1$}} \ptbeg \ptnode{\hspace{-1mm}$\stackrel{\neg S(2,3) {\Large \ \surd}}
{\not\exists \nit{DT} \mbox{ for } S(2,3)?}$} \ptleaf{\fbox{\sf succeeds} \ignore{\ptleaf{\hspace{5mm}$T(4,1,2)\!:\!1$}} } \ptend \ptend 
\ptend
\ptendtree
}

{\tiny
\hspace*{-0.8cm}\ptbegtree
\ptbeg
\ptnode{\centerline{$P(1,2)$}}
\ptbeg
\ptnode{\fbox{$\rho_1$}} \ptbeg \ptnode{$R(1,2)$\hspace{5mm}}\ptbeg \ptnode{\fbox{$\rho_2$}}  \ptleaf{$Q(1,2,3)\!:\!1$\hspace{5mm}} \ptend \ptend \ptend
\ptbeg
\ptnode{\fbox{$\rho_1$}} \ptbeg \ptnode{\hspace{5mm}$\stackrel{\hspace*{-6mm}\neg S(1,2) {\Large \ \times}}
{\hspace*{-6mm}\not\exists \nit{DT} \mbox{ for } S(1,2)?}$}
   \ptleaf{\fbox{\sf fails} \ignore{\ptleaf{\hspace{5mm}$T(4,1,2)\!:\!1$}} } \ptend \ptend 
\ptend
\ptendtree
}

{\tiny
\hspace*{-0.8cm}\ptbegtree
\ptbeg
\ptnode{\centerline{$P(1,2)$}}
\ptbeg
\ptnode{\fbox{$\rho_1$}} \ptbeg \ptnode{$R(1,2)$\hspace{5mm}}\ptbeg \ptnode{\fbox{$\rho_2$}}  \ptleaf{$Q(1,2,5)\!:\!1$\hspace{5mm}} \ptend \ptend \ptend
\ptbeg
\ptnode{\fbox{$\rho_1$}} \ptbeg \ptnode{\hspace{5mm}$\stackrel{\hspace*{-6mm}\neg S(1,2)  {\Large \ \times}}
{\hspace*{-6mm}\not\exists \nit{DT} \mbox{ for } S(1,2)?}$} \ptleaf{\fbox{\sf fails} \ignore{\ptleaf{\hspace{5mm}$T(4,1,2)\!:\!1$}} } \ptend \ptend 
\ptend
\ptendtree
}
\end{multicols}
\caption{Derivation-trees for $P(2,3)$ and $P(1,2)$ in Example \ref{ex:diff}.}
\label{fig:derivation-trees-example-three}
\end{figure}

\subsection{Warded Datalog$^{\pm}$}\label{sec:wardedD}
\label{sec:basics}

Datalog$^{\pm}$ was introduced in \cite{DBLP:conf/icdt/CaliGL09} as an extension of Datalog, where the ``$+$" stands for the new ingredients, namely:
tuple-generating dependencies ({\em tgds\/}), written as {\em existential rules} of the form
$\exists \bar{z} P(\bar{x}',\bar{z})  \ \leftarrow \ P_1(\bar{x}_1), \ldots, P_n(\bar{x}_n)$,
with $\bar{x}' \subseteq \bigcup_i\bar{x}_i$, and $\bar{z} \cap \bigcup_i\bar{x}_i = \emptyset$;
as well as equality-generating dependencies ({\em egds\/}) and negative constraints.
In this work we ignore egds and constraints.\ignore{ and we restrict tgds to a single atom in the head. We refer to these
tgds as {\em existential rules\/} (or, simply, {\em rules\/}).}
The ``$-$`` in Datalog$^{\pm}$ stands for syntactic restrictions on \ignore{the syntax of the existential} rules
to ensure decidability of CQ answering\ignore{ over a database and a Datalog$^{\pm}$ program}.

We consider three sets of term symbols: $\mathbf{C}$, $\mathbf{B}$, and $\mathbf{V}$ containing constants, labelled nulls
(also known as blank nodes in the semantic web context), and variables, respectively.
Let $T$ denote an atom or a set of atoms. We write
$\var(T)$ and $\nulls(T)$ to denote the set of variables and nulls, respectively, occurring in $T$.
In a DB, typically an EDB $D$, all terms are from $\mathbf{C}$.
In an {\em instance}, we also allow terms to be from $\mathbf{B}$.
For a rule $\rho$, $\nit{body}(\rho)$ denotes the set of atoms in its body, and
$\nit{head}(\rho)$, the atom in its head.
A homomorphism $h$ from a set $X$ of atoms to a set $X'$ of atoms is a partial function
$h \colon \mathbf{C} \cup \mathbf{B} \cup \mathbf{V} \rightarrow \mathbf{C} \cup \mathbf{B} \cup \mathbf{V}$ such that
$h(c) = c$ for all $c \in \mathbf{C}$ and
$P(h(t_1), \dots, h(t_k)) \in X'$
for every atom $P(t_1, \dots, t_k) \in X$.


We say that a rule $\rho \in \Pi$ is {\em applicable\/} to an instance $I$ if there exists a homomorphism $h$ from
$\body(\rho)$ to $I$. In this case, the result of applying $\rho$ to $I$ is an instance $I' = I \cup \{h'(\head(\rho))\}$,
where $h'$ coincides with $h$ on $\var(\body(\rho))$ and $h'$ maps each existential variable in the head of $\rho$
to a fresh labelled null not occurring in $I$. For such an application of $\rho$ to $I$, we write $I \langle \rho,h\rangle I'$.
Such an application of $\rho$ to $I$ is called a {\em chase step\/}. The chase
is an important tool in the presence of existential rules. A {\em chase sequence\/} for a DB $D$ and a
\dpm program $\Pi$ is a sequence of chase steps $I_i \langle \rho_i, h_i\rangle I_{i+1}$ with $i \geq 0$, such that
$I_0 = D$ and  $\rho_i \in \Pi$ for every $i$ (also denoted $I_i \leadsto_{\rho_i,h_i}I_{i+1}$).
For the sake of readability, we sometimes only write the newly generated atoms
of each chase step without repeating the old atoms. Also the subscript $\rho_i,h_i$ is omitted if it is
clear
from the context. A chase sequence then reads $I_0   \leadsto A_1 \leadsto A_2 \leadsto A_3 \leadsto \dots$
with $I_{i+1} \setminus I_i = \{ A_{i+1}\}$.

The final atoms of all possible chase sequences
for DB $D$ and $\Pi$ form an instance referred to as $\nit{Chase}(D,\Pi)$, which
can be infinite.
We denote the result of all chase sequences up to length $d$ for some  $d \geq 0$ as $\nit{Chase}^d(D,\Pi)$.
The chase variant assumed here is the so-called {\em oblivious chase}
\cite{DBLP:journals/jair/CaliGK13,DBLP:journals/jcss/JohnsonK84}, i.e.,
if a rule $\rho \in \Pi$ ever becomes applicable with some homomorphism $h$, then
$\nit{Chase}(D,\Pi)$ contains exactly one atom of the form
$h'(\head(\rho))$ such that $h'$ extends $h$ to
the existential variables in the head of $\rho$. Intuitively,  each rule is applied exactly once
for every applicable homomorphism.


Consider a DB $D$ and a \dpm program $\Pi$ (the former the EDB for the latter). As a logical theory, $\Pi \cup D$ may have multiple models, but
the model $M = \nit{Chase}(D,\Pi)$ turns out to be a correct representative for the class of models: for every
BCQ $\mc{Q}$, $\Pi \cup D \models \mc{Q}$ iff $M \models Q$
\cite{DBLP:journals/tcs/FaginKMP05}.
\ignore{As mentioned before, the chase instance may be infinite. However,} There are classes of Datalog$^\pm$
that, even with an infinite chase, allow for decidable or even tractable CQA in the size of the EDB. Much effort has been made in identifying and characterizing interesting syntactic classes of programs with this property (see \cite{DBLP:journals/jair/CaliGK13} for an overview)\ignore{ of decidability paradigms in the
\dpm family)}. In this direction,
{\em warded} Datalog$^\pm$  was introduced in \cite{DBLP:conf/pods/ArenasGP14,DBLP:conf/pods/ArenasGP18,DBLP:conf/ijcai/GottlobP15},
as a particularly well-behaved fragment of Datalog$^{\pm}$, for which CQA is tractable.
We briefly recall and illustrate it here, for which we need some preliminary notions.


A position $p$ in Datalog$^\pm$  program $\Pi$ is {\em affected} if:
(a)
an existential variable appears in $p$, or
(b) there is $\rho \in \Pi$ such that a variable $x$ appears in $p$ in $\nit{head}(\rho)$ and all occurrences of $x$ in $\nit{body}(\rho)$ are in affected positions.
$\nit{Aff}(\Pi)$ and $\nit{NonAff}(\Pi)$ denote the sets of affected, resp. non-affected, positions of  $\Pi$.
Intuitively, $\nit{Aff}(\Pi)$ contains all positions where the chase
may possibly introduce~a~null.

\begin{example} \label{ex:Aff} Consider the following program:
\begin{align}
\exists z_1 \ R(y_1,z_1) \! ~&\leftarrow~ \! P(x_1,y_1)\label{eq:edg1}\\
\exists z_2 \ P(x_2,z_2) \! ~&\leftarrow~ \! S(u_2,x_2,x_2), R(u_2,y_2)\label{eq:edg2}\\
\exists z_3 \ S(x_3,y_3,z_3) \! ~&\leftarrow~ \! P(x_3,y_3), U(x_3).\label{eq:edg3}
\end{align}
By the first case, $R[2],P[2],S[3]$ are affected.
By the second case, $R[1],S[2]$ are affected. Now that $S[2], S[3]$ are affected, also $P[1]$ is.
We thus have
$\nit{Aff}(\Pi)$ $ =$ $\{P[1],P[2],R[1], R[2]$, $S[2],S[3] \}$ and
$\nit{NonAff}(\Pi)$ $=$ $\{S[1],$ $U[1]\}$. \boxtheorem
\end{example}

For a rule $\rho \in \Pi$, and a variable $x$:  (a)
$x \in \nit{body}(\rho)$ is {\em harmless} if it appears at least once in $\nit{body}(\rho)$ at a position in $\nit{NonAff}(\Pi)$. \
$\nit{Harmless}(\rho)$ denotes the set of harmless variables in $\rho$. Otherwise, a variable is called {\em harmful}.
Intuitively, harmless variables will always be instantiated to constants in any chase step, whereas harmful variables
may be instantiated to nulls.
(b)
$x \in \nit{body}(\rho)$ is {\em dangerous} if $x \notin \nit{Harmless}(\rho)$ and $x \in \nit{head}(\rho)$. \ $\nit{Dang}(\rho)$ denotes the set of dangerous  variables~in~ $\rho$. These are the variables which may propagate nulls into the atom created by a chase step.

 \begin{example} \ (ex.\ \ref{ex:Aff}  cont.) \label{ex:Aff+} \
$x_1$ and $y_1$ are both harmful but only $y_1$ is dangerous for (\ref{eq:edg1}); \
$u_2$ is harmless, $x_2$ is dangerous, $y_2$ is harmful but not dangerous for (\ref{eq:edg2}); \
$x_3$ is harmless and $y_3$ is dangerous for (\ref{eq:edg3}).
 \boxtheorem
 \end{example}
Now,  a rule $\rho \in \Pi$ is {\em warded\/} if either $\nit{Dang}(\rho) = \emptyset$ or there exists
an atom $A \in \body(\rho)$, the ward, such that
(1) $\nit{Dang}(\rho) \subseteq \var(A)$ and (2)
$\var(A) \cap \var(\body(\rho) \setminus \{A\}) \subseteq \nit{Harmless}(\rho)$.
A program  $\Pi$ is {\em warded} if every rule $\rho \in \Pi$ is warded.

 \begin{example} \ (ex.\ \ref{ex:Aff+} cont.) \
Rule  (\ref{eq:edg1}) is trivially warded with the single body atom as the ward.
Rule  (\ref{eq:edg2}) is warded by $S(u_2,x_2,x_2)$:
variable $x_2$ is the only dangerous variable
and $u_2$ (the only variable shared by the ward with the rest of the body) is harmless. Actually,
the other body atom $R(u_2,y_2)$ contains the harmful variable $y_2$;
but it is not dangerous and not shared with
the ward.
Finally, rule (\ref{eq:edg3}) is warded by $P(x_3,y_3)$; the other atom $U(x_3)$ contains no affected variable.
Since all rules are warded, the program $\Pi$ is warded. \boxtheorem
\end{example}

\dpm can be extended with safe, stratified negation in the usual way,
 similarly as stratified Datalog \cite{AHV94}. The resulting Datalog$^{\pm,\neg s}$ can also be given a chase-based semantics \cite{DBLP:journals/ws/CaliGL12}.
The notions of affected/non-affected positions and harmless/harmful/dangerous variables carry over to a Datalog$^{\pm,\neg s}$ program $\Pi$ by considering only the program $\Pi^{\rm pos}$ obtained
from $\Pi$ by deleting all negated body atoms. For warded \dpm, only a restricted form of stratified negation is
allowed -- so-called {\em stratified ground\/} negation. This means that we require for every rule $\rho \in \Pi$:
if $\rho$ contains a negated atom $P(t_1,\dots,t_n)$, then every $t_i$ must be either a constant (i.e, $t_i \in \mathbf{C}$) or
\ignore{$t_i$} a harmless variable. Hence, negated atoms can never contain a null in the chase. We write \dpmsn for programs in this language.

The class of warded \dpmsn  programs extends the class of Datalog$^{\neg s}$ programs.
Warded \dpmsn is expressive enough to capture query answering of a core fragment of SPARQL under certain OWL 2 QL entailment regimes
\cite{glimm},
and this task can actually be accomplished in polynomial time
(data complexity)
\cite{DBLP:conf/pods/ArenasGP14,DBLP:conf/pods/ArenasGP18}.
Hence, \dpmsn constitutes a very good compromise between expressive
power and complexity. Recently, a powerful query engine for
warded \dpmsn has been presented \cite{DBLP:conf/ijcai/BellomariniGPS17,pvldb/BellomariniSG18},
namely the VADALOG system.


\section{Datalog$^\pm$-Based Bag Semantics for Datalog$^{\neg s}$}
\label{sec:datalog}

We now provide a set-semantics that represents the bag semantics for a Datalog program $\Pi$ with a multiset EDB $D$ via the transformation into a Datalog$^\pm$  program $\Pi^+$ over
a set EDB $D^+$ obtained from $D$. For this, we assume w.l.o.g., that the set of {\em nulls}, $\mathbf{B}$, for a Datalog$^{\pm}$ program is partitioned into two infinite ordered sets $\mc{I} = \{\iota_1,\iota_2, \ldots\}$, for unique, global tuple identifiers (tids), and $\mc{N} = \{\eta_1,\eta_2, \ldots\}$, for usual nulls in Datalog$^\pm$  programs.
Given a multiset EDB $D$ and a program $\Pi$, instead of using colors
and syntactically different derivation trees, we will  use elements of $\mc{I}$ to identify both the elements of the EDB and the tuples resulting from  applications of the chase.

For every predicate $P(\ldots)$, we introduce a new version $P(. \ ;\ldots)$ with an extra, first argument (its $0$-th position) to accommodate a tid, which is a null from $\mc{I}$. If an atom $P(\bar{a})$ appears in $D$ as
$n$ duplicates, we create the tuples \ $P(\iota_1';\bar{a}), \ldots, P(\iota_n';\bar{a})$, with the $\iota_i'$ pairwise different nulls from $\mc{I}$ as tids, and not used to identify any other element of $D$.
We obtain a {\em set} EDB $D^+$ from the multiset EDB $D$.
%
Given a  rule in $\Pi$, we introduce tid-variables (i.e. appearing in the $0$-th positions of predicates) and existential quantifiers in the rule head,  to  formally generate fresh tids when the rule applies. More precisely, a rule in $\Pi$ of the form \ $\rho\!: \ H(\bar{x}) \ \leftarrow \ A_1(\bar{x}_1),A_2(\bar{x}_2),\ldots,A_k(\bar{x}_k)$, with $k>0, \ \bar{x} \subseteq \cup_i \bar{x}_i$, becomes the Datalog$^\pm$ rule \ $\rho^+\!: \ \exists z \ H(z;\bar{x}) \leftarrow A_1(z_1;\bar{x}_1),A_2(z_2;\bar{x}_2),\ldots,A_k(z_k;\bar{x}_k)$, with fresh, different variables $z, z_1, \ldots,z_k$.
The resulting Datalog$^\pm$  program $\Pi^+$ can be evaluated according to
the usual {\em set semantics\/} on the set EDB $D^+$ via  the  chase:
when the instantiated body $A_1(\iota_1';\bar{a}_1)$,
$A_2(\iota_2';\bar{a}_2),\ldots,A_k(\iota_k';\bar{a}_k)$ of rule $\rho^+$
becomes true, then the new tuple $H(\iota;\bar{a})$ is created,
with $\iota$ the first (new) null from $\mc{I}$ that has not been used yet, i.e., the tid of the new atom.
\ignore{****************
The chase-sequences with $\Pi^+$ and $D^+$ allow us to define by transitivity
a relationship, $D^+ \leadsto^* A$, between instance $D^+$ and an  atom $A$ in the chase instance.
We will assume below that chase sequences are {\em minimal} in that every intermediate atom derived
is used to enforce some \tgd\ later along the same sequence.
****************}

\begin{example} \! \label{ex:theex+} (ex.\ \ref{ex:theex} cont.) \ The EDB $D$ from
Example \ref{ex:theex} becomes

\smallskip

$D^+ =
\{Q(\iota_1;1,2,3)$,
$Q(\iota_2;1,2,5), Q(\iota_3;2,3,4), Q(\iota_4;2,3,4), T(\iota_5;4,1,2), T(\iota_6;4,1,2) \}$;%
\footnote{Notice that this set version of $D$ can also be created by means of Datalog$^\pm$ rules. For example, with the rule \ $\exists zQ(z;x,y,v) \leftarrow Q(x,y,v)$ for the EDB predicate $Q$.}

\smallskip

\noindent
and program
$\Pi$ becomes
$\Pi^+ = \{\rho_1^+,\rho_2^+,\rho_3^+ \}$ with
$\rho_1^+\!: \ \exists u P(u;x,y)  \leftarrow R(u_1;x,y), S(u_2;x,y)$; \
$\rho_2^+\!: \ \exists u R(u;x,y) \leftarrow Q(u_1;x,y,z)$; \
$\rho_3^+\!: \ \exists u S(u;x,y) \leftarrow T(u_1;z,x,y).$

\smallskip

The following is a 3-step chase sequence of $D^+$ and $\Pi^+$:
$\{Q(\iota_1;1,2,3), T(\iota_4;4,1,2)\}$
$\leadsto_{\rho_2^+}$ $R(\iota_7;1,2)$
$\leadsto_{\rho_3^+}$  $S(\iota_{11};1,2)$
$\leadsto_{\rho_1^+}$
$P(\iota_{13};1,2)$.

\smallskip

Analogously to the depth-two and depth-three trees in Example~\ref{ex:theex},
the chase produces 10 new atoms. In total, we get: \
$\nit{Chase}(\Pi^+,D^+)$
$ = D^+ \cup
\{R(\iota_7;1,2),$ $ R(\iota_8;1,2),$ $R(\iota_9;2,3),$ $R(\iota_{10};2,3),$ $S(\iota_{11};1,2)$,
$S(\iota_{12};1,2),$  $ P(\iota_{13};1,2),$  $P(\iota_{14};1,2),$ $P(\iota_{15};1,2),$ $ P(\iota_{16};1,2)\}$.
\boxtheorem
\end{example}

In order to extend the PBB Semantics to Datalog$^{\neg s}$, we have to extend our
transformation of programs $\Pi$ into $\Pi^+$ to rules with negated atoms.
Consider a rule $\rho$ of the form:
\begin{equation}
\rho\!: \  H(\bar{x}) \ \leftarrow \ A_1(\bar{x}_1), \ldots, A_k(\bar{x}_k),
\neg B_1(\bar{x}_{k+1}), \ldots, \neg B_\ell(\bar{x}_{k+\ell}), \label{eq:forRew}
\end{equation}

\noindent with, $\bar{x}_{k+1}, \ldots, \bar{x}_{k+\ell}  \subseteq \bigcup_{i=1}^k \bar{x}_i$;
we transform it
 into the following two rules:
 \begin{eqnarray}
\rho^+\!: \  \exists z H(z;\bar{x}) &~\leftarrow~& A_1(z_1;\bar{x}_1), \ldots, A_k(z_k;\bar{x}_k),\nonumber 
\neg \nit{Aux}_1^\rho(\bar{x}_{k+1}), \ldots, \neg \nit{Aux}_\ell^\rho(\bar{x}_{k+\ell}),\\
\nit{Aux}_i^\rho(\bar{x}_{k +i}) &~\leftarrow~& B_i(z_i;\bar{x}_{k+i}), \ \ i = 1, \ldots, \ell. \label{eq:aftRew}
\end{eqnarray}
The introduction of auxiliary predicates $\nit{Aux}_i$ is crucial since adding fresh variables
directly to the negated atoms  would yield negated atoms of the form
$\neg B_i(z_{k+i};\bar{x}_{k+i})$ in the rule body, which make the rule unsafe.
The resulting Datalog$^{\pm, \neg s}$ program is from the desired class
\dpmsn:

\begin{theorem}
\label{theorem:piplus-is-warded} \em
Let $\Pi$ be a Datalog$^{\neg s}$ program and let $\Pi^+$ be the transformation of $\Pi$ into
a Datalog$^{\pm, \neg s}$ program. Then, $\Pi^+$ is a warded \dpmsn program.
\end{theorem}

\ignore{\begin{proof}
It is easy to verify that, in every $\rho^+ \in \Pi^+$, the only harmful position is the
position 0 (i.e, the tid) of each predicate from $\Pi$. However, the variables occurring in position 0
in the rule bodies do not occur in the head. Hence, none of the rules in $\Pi^+$ contains a dangerous variable and,
therefore, $\Pi^+$ is trivially warded. By the same consideration, the auxiliary predicates introduced in
(\ref{eq:aftRew}) above contain only harmless positions. Since these are the only negated atoms in rules of $\Pi^+$,
we have only ground negation in $\Pi^+$. Finally, the stratification of $\Pi$ carries over to $\Pi^+$, where
an auxiliary predicate
$\nit{Aux}_i^\rho$ introduced for a negated predicate $B_i$ in the definition of predicate  $H$ by $\rho$ in $\Pi$ ends up in stratum $s$ of $\Pi^+$ iff $B_i$ is
in stratum $s$ in $\Pi$ (c.f. (\ref{eq:forRew}) and (\ref{eq:aftRew})).
\end{proof} }

Operation $\nit{col}^{-1}$ of
Section \ref{sec:preliminaries} inspires de-identification  and multiset merging operations. Sometimes we use double braces, $\{\!\{\ldots\}\!\}$, to emphasize that the operation produces a {\em multiset\/}.

\begin{definition}
\label{def:di-and-sp}
For a set $D$ of tuples with tids,
$\mc{DI}(D)$ and $\mc{SP}(D)$, for de-identification and set-projection, respectively,
are:  (a) $\mc{DI}(D) := \{\!\{P(\bar{c})~|$ $P(t;\bar{c}) \in D \mbox{ for some } t\}\!\}$,
a {\it multiset\/}; \ and \ (b)  $\mc{SP}(D) := \{P(\bar{c})~|$ $ P(t;\bar{c}) \in D$ $
\mbox{ for some } t\}$, a {\it set\/}.
 \boxtheorem
\end{definition}

\begin{definition}\label{def:sem1}
Given a Datalog$^{\neg s}$ program $\Pi$
and a multiset EDB $D$, the {\em program-based bag semantics} (PBB semantics) assigns to $\Pi \cup D$ the multiset:

\smallskip

$\nit{PBBS}(\Pi,D) := \mc{DI}(\nit{Chase}(\Pi^+,D^+))=
\{\!\{P(\bar{a})~|~ P(t;\bar{a}) \in \nit{Chase}(\Pi^+,D^+)\}\!\}.$
\boxtheorem
\end{definition}

The main results in this section are the correspondence of
PBB semantics  and DTB  semantics and the relationship of both with classical set semantics of Datalog:

\begin{theorem} \label{thm:isom}
\em
For a Datalog$^{\neg s}$ program $\Pi$  with a multiset EDB $D$,
\ $\nit{DTBS}(\Pi,D) = \nit{PBBS}(\Pi,D)$ holds.
\end{theorem}

\begin{proof}[Proof Idea]
The theorem is proved by establishing a one-to-one correspondence between DTs in $\nit{DT}(\Pi,\nit{col}(D))$ with a fixed
root atom $P(\bar{c})$ and (minimal) chase-derivations of $P(\bar{c})$ from $D^+$ via $\Pi^+$.
This proof proceeds by induction on the
depth of the DTs and length of the chase sequences.
\end{proof}

\begin{corollary}\label{cor:model}
 \em
Given a Datalog (resp.\  Datalog$^{\neg s}$) program $\Pi$ and a multiset EDB $D$,
the set \linebreak
$\mc{SP}(\nit{Chase}(\Pi^+,D^+))$ is the minimal model (resp.\ the standard model) of the
program $\Pi \cup  \mc{SP}(D)$.
\end{corollary}


\section{Bag Semantics for Datalog$^{\pm,\neg sg}$}
\label{sec:warded}

In the previous section, we have seen that warded \dpm (possibly extended with stratified ground negation)
is well suited to capture the bag semantics of Datalog in terms of classical set semantics.
We now want to go one step further and study the {\em bag semantics\/} for \dpmsn itself. Note that this question makes a lot of sense given the results from \cite{DBLP:conf/pods/ArenasGP18}, where it is shown that warded
\dpmsn captures a core fragment of SPARQL under certain
OWL2 QL entailment regimes
and  the official W3C semantics of SPARQL is a bag semantics.

\subsection{Extension of the DTB Semantics to Datalog$^{\pm, \neg sg}$}
\label{sec:dpm-dtb}

The definition of a DT-based bag semantics for \dpmsn is not as straightforward as for Datalog$^{\neg s}$,
since atoms in a model of $D \cup \Pi$ for a (multiset) EDB $D$ and
\dpmsn program $\Pi$ may have labelled nulls as arguments, which
 correspond to existentially quantified variables and may be arbitrarily chosen.
Hence, when counting duplicates, it is not clear whether two atoms differing only in their choice of nulls
should be treated as copies of each other or not.
We therefore treat multiplicities of atoms analogously to
multiple answers to single-atom queries, i.e., the multiplicity of an atom $P(t_1,  \dots, t_k)$ wrt.\ EDB $D$ and
program $\Pi$ corresponds to the multiplicity of answer $(t_1,  \dots, t_k)$ to the query
$Q = P(x_1, \dots, x_k)$ over the database $D$ and program $\Pi$. In other words,
we ask for all instantiations $(t_1, \dots, t_k)$ of $(x_1, \dots, x_k)$ such that
$P(t_1, \dots, t_k)$ is true in {\em every} model of $D \cup \Pi$.
It is well known that only ground instantiations (on $\mathbf{C}$)
can have this property (see e.g.\ \cite{DBLP:journals/tcs/FaginKMP05}).
Hence, below, we restrict ourselves to considering duplicates of ground atoms containing constants from $\mathbf{C}$ only (in this section we are not using tid-positions $0$). In the rest of this section, unless otherwise stated, ``ground atom" means instantiated on $\mathbf{C}$; and programs belong to \dpmsn.

In order to define the  multiplicity of a ground atom $P(t_1, \dots, t_k)$
wrt.\ a (multiset) EDB $D$ and a warded \dpmsn program $\Pi$, we adopt the notion of {\em proof tree} used in
\cite{DBLP:conf/pods/ArenasGP18,DBLP:journals/ai/CaliGP12}, which generalizes the notion of derivation tree
to \dpmsn.
We consider first positive \dpm programs. A proof tree (PT) for an atom $A$ (possibly with nulls) wrt.\ (a set) EDB $D$ and \dpm program $\Pi$
is a  node- and edge-labelled tree with labelling function $\lambda$,
such that: (1) The nodes are labelled by atoms over $\mathbf{C} \cup \mathbf{B}$. \
(2) The edges are labelled by rules from $\Pi$. \
(3) The root is labelled by $A$. \ (4) The leaf nodes are labelled by atoms from $D$. \
(5) The edges from a node $N$ to its child nodes $N_1, \dots, N_k$ are
all labelled with the same rule $\rho$. \ (6) The atom labelling $N$ corresponds to the result of
a chase step where $\body(\rho)$ is instantiated to $\{\lambda(N_1), \dots, \lambda(N_k)\}$
and $\head(\rho)$ becomes $\lambda(N)$ when instantiating the existential variables of $\head(\rho)$
with fresh nulls. \ (7) If $M'$ (resp.\ $N'$) is the parent node of $M$ (resp.\ $N$) such that
$\nulls(\lambda(M')) \setminus \nulls(\lambda(M))$  and $\nulls(\lambda(N')) \setminus \nulls(\lambda(N))$
share at least one null, then
the entire subtrees rooted at $M'$ and at $N'$ must be isomorphic (i.e., the same
tree structure and the same labels). \
(8) If, for two nodes $M$ and $N$, $\lambda(M)$ and $\lambda(N)$ share a null $v$,
then there exist ancestors $M'$ of $M$ and $N'$ of $N$ such that
$M'$ and $N'$ are siblings corresponding to two body atoms $A$ and $B$ of rule $\rho$ with
$x \in \var(A)$ and $x\in \var(B)$ for some variable $x$ and $\rho$ is applied with some substitution $\gamma$ which sets
$\gamma(x)=v$; moreover, $v$ occurs in the labels of the entire paths from $M'$ to $M$ and from $N'$ to $N$.
A proof tree for Example \ref{ex:proof-tree} below is shown in Figure \ref{fig:proof-tree}, left.
As with derivation trees, we assume that siblings  in the proof tree are arranged in the order of the corresponding body atoms
in the rule  $\rho$ labelling the edge to the parent node
(cf.\  Section \ref{sec:dtbs}).

Intuitively, a PT is a tree-representation of the derivation of an atom by the chase. The parent/child
relationship in a PT corresponds to the head/body of a rule application in the chase. Condition (7) above
refers to the case that a non-ground atom is used in two different rule applications within the chase sequence.
In this case, the two occurrences of this atom  must have identical proof sub-trees. A PT can be obtained from a chase-derivation by
{\em reversing the edges and unfolding the chase graph into a tree by copying some of the nodes\/} \cite{DBLP:conf/pods/ArenasGP18}.
By definition of the chase in Section \ref{sec:basics},
it can never happen that the same null is created
by two different chase steps. Note that the nulls in $\nulls(\lambda(M')) \setminus \nulls(\lambda(M))$ (and, likewise in
$\nulls(\lambda(N')) \setminus \nulls((\lambda(N))$) are precisely the newly created ones. Hence, if
$\lambda(M')$ and $\lambda(N')$ share such a null, then $\lambda(M')$ and $\lambda(N')$ are the same atom and
the subtrees rooted at these nodes are obtained by unfolding the same subgraph of the chase graph.
Condition (8) makes sure that we use the same null $v$ in a PT only if this is required by a join condition of some
rule $\rho$; otherwise nulls are renamed apart.

\begin{figure}
\begin{center}
\includegraphics[height=55.0mm]{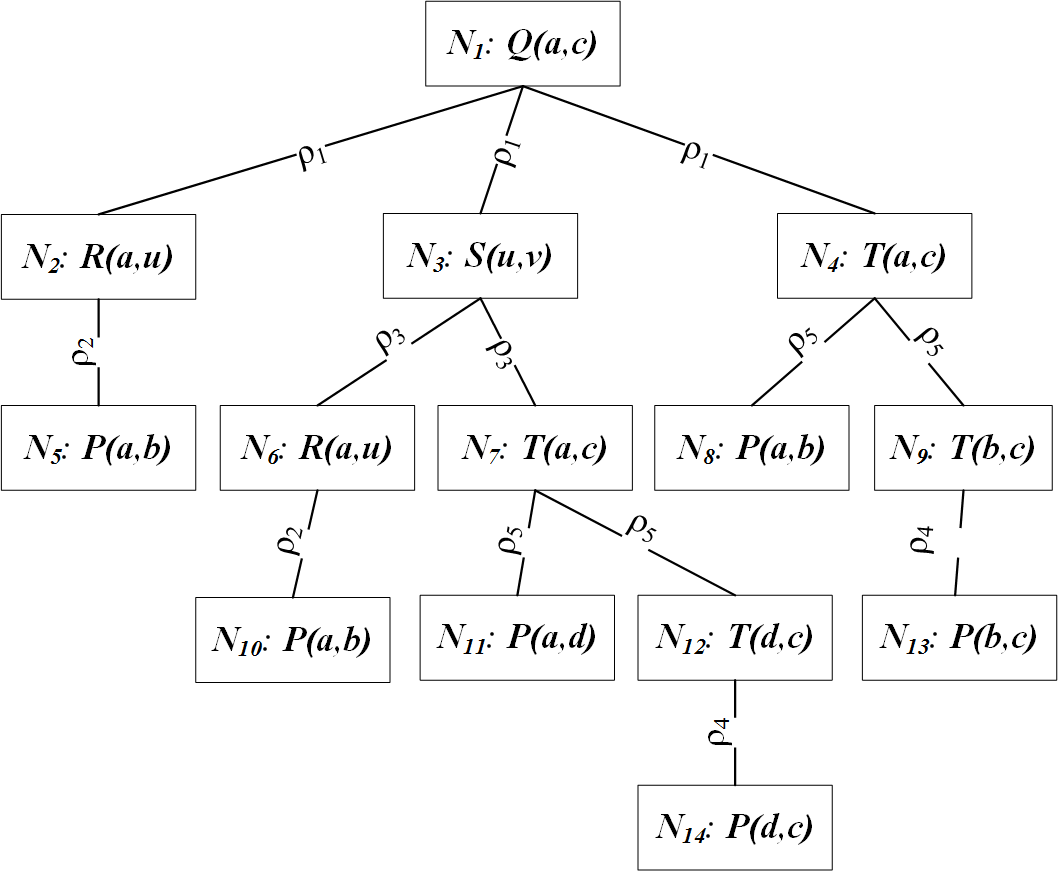}
\quad \quad
\includegraphics[height=55.0mm]{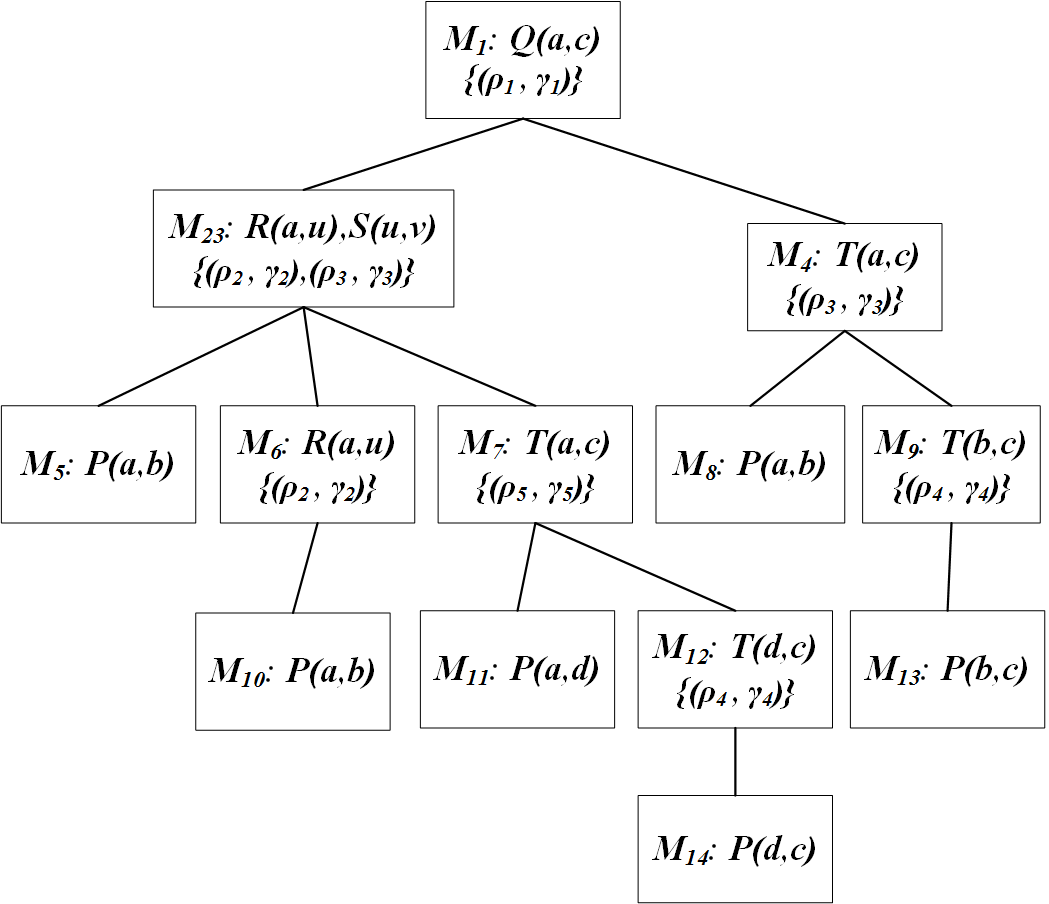}\vspace{-4mm}
\end{center}
\caption{Proof tree (left) and witness (right) for atom $Q(a,c)$ in Examples \ref{ex:proof-tree} and \ref{ex:witness}.}
\label{fig:proof-tree}
\end{figure}

\begin{example}
\label{ex:proof-tree}
Let $\Pi = \{\rho_1, \dots, \rho_5\}$ be the  \dpm program with
$\rho_1 \!: \ Q(x,w) \leftarrow R(x,y),$ $ S(y,z), T(x,w)$; \
$\rho_2 \!: \ \exists z~R(x,z) \leftarrow P(x,y)$; \
$\rho_3 \!: \ \exists z~S(x,z) \leftarrow R(w,x), T(w,y)$; \
$\rho_4 \!: \ T(x,y) \leftarrow P(x,y)$; \
$\rho_5 \!: \ T(x,y) \leftarrow P(x,z), T(z,y)$.
This program belongs to \dpmsn,
\linebreak
and -- although not necessary to build a proof tree for it -- we notice that it is also warded: $\nit{Aff}(\Pi) = \{R[2], S[2], S[1]\}$; and all other positions are  not affected. Rule $\rho_3$ is warded with ward $R(w,x)$ (where
$x$ is the only dangerous variable in this rule). All other rules are trivially warded because they have no dangerous variables.

Now let $D = \{ P(a,b), P(b,c), P(a,d), P(d,c)\}$. A possible proof tree for $Q(a,c)$ is shown in Figure \ref{fig:proof-tree} on the left.
It is important to note that nodes $N_2$ and $N_6$ introducing labelled null  $u$ are labelled with the same atom and give rise to identical subtrees. Of course, $R(a,u)$ could also result from a chase step applying $\rho_2$ to $P(a,d)$. However, this would generate
a null different from $u$ and, subsequently, the nulls in $R(a,\_)$ and $S(\_,v)$ (in $N_2$ and $N_3$)
would be different, and rule $\rho_1$ could not be applied
anymore.
In contrast, the nodes $N_4$ and $N_7$ with label $T(a,c)$ span different subtrees.
This is desired: there are two possible derivations for each occurrence of atom $T(a,c)$ in the PT.
\boxtheorem
\end{example}

Here we deviate slightly from the definition of PTs in \cite{DBLP:conf/pods/ArenasGP18}, in that we allow the same {\em ground\/} atom to have different derivations. This is needed to detect duplicates and to
make sure that PTs in fact constitute a generalization of the derivation trees in
Section \ref{sec:dtbs}. Moreover, condition (8) is needed to avoid non-isomorphic PTs by
``unforced''
repetition of nulls (i.e., identical nulls that are not required by a join condition
further up in the tree).
Analogous to the generalization  in Section \ref{sec:dtbs} of DTs for Datalog to DTs for Datalog$^{\neg s}$, it is easy to generalize proof trees to \dpmsn.
Here it is important that we only allow
stratified {\em ground\/} negation. Hence, analogously to DTs for Datalog$^{\neg s}$, we
allow negated ground atoms $\neg A$ to occur as node labels of leaf nodes in a PT, provided
that the positive atom $A$  has no PT. Moreover, it is no problem to allow also
multiset EDBs since, as in Section~\ref{sec:dtbs}, we can keep duplicates apart by
means of a coloring function
\nit{col}.

Finally, we can define proof trees $\calT$ and
$\calT'$ as equivalent, denoted  $\calT \equiv \calT'$, if one is obtained from the other by renaming of nulls.
We can thus normalize PTs by assuming that nulls in node labels are from some fixed set $\{u_1, u_2\, \dots\}$ and
that these nulls are introduced in the labels of the PT by some fixed-order traversal
(e.g., top-down, left-to-right).

For a PT $\calT$, we define $\nit{Atoms}(\calT)$
as $\nit{col}^{-1}$ of the root when $\calT$ is a single-node tree,
and the root-label of  $\calT$, otherwise.
For a set of PTs $\mf{T}$: \ $\nit{Atoms}(\mf{T}) := \biguplus \{\nit{Atoms}(\calT)~|~\calT\in \mf{T}\}$, which is a multiset that multi-contains $D$. \
If $\nit{PT}(\Pi,$ $\nit{col}(D))$ is the set of normalized, {\em non-equivalent\/} PTs,  the
{\em proof-tree bag (PTB) semantics} for $\Pi$ is the multiset
$\nit{PTBS}(\Pi,D) := \nit{Atoms}(\nit{PT}(\Pi,\nit{col}(D)))$.
For a ground atom $A$,  $\nit{mult}(A,\nit{PTBS}(\Pi,D))$ denotes the multiplicity of $A$
 in the multiset $\nit{PTBS}(\Pi,D)$.

\ignore{
\begin{equation}
\nit{PTBS}(\Pi,D) := \nit{Atoms}(\nit{PT}(\Pi,\nit{col}(P))). \label{eq:dtbs-stratified}
\end{equation}
}

\begin{example}
\label{ex:ptb-semantics}
(ex.\ \ref{ex:proof-tree} cont.)
To compute $\nit{PTBS}(\Pi,D)$  for $\Pi$ and $D$ from Example \ref{ex:proof-tree}, we have to determine all proof trees of all
ground atoms derivable from $\Pi$ and $D$. In Figure \ref{fig:proof-tree}, we have already seen one proof tree for $Q(a,c)$;
and we have observed that the sub-proof tree of $R(a,u)$ rooted at nodes $N_2$ and $N_6$ could be replaced by a child node with label $P(a,d)$ (either both subtrees or none has to be changed).
In total, the ground atom $Q(a,c)$ has 8 different proof trees wrt.\ $\Pi$ and $D$ (multiplying the 2 possible derivations of atom $R(a,u)$ with 4 derivations for the two occurrences of the atom $T(a,c)$ in nodes $N_4$ and $N_7$).
The other ground atoms derivable from $\Pi$ and $D$ are $T(a,c)$ (with 2 possible PTs as discussed in
Example \ref{ex:proof-tree})
and the atoms $T(i,j)$ for each
$P(i,j)$ in $D$. Hence, we have $\nit{PTBS}(\Pi,D) = D \cup \{T(a,b),T(b,c),T(a,d),T(d,a),T(a,c),T(a,c), Q(a,c), \dots, Q(a,c)\}$.
\boxtheorem
\end{example}

Clearly, every Datalog$^{\neg s}$ program is a special case of a warded \dpmsn program.
It is easy to verify that the PTB semantics indeed generalizes the DTB semantics:

\begin{proposition}
\label{prop:DTBS-vs-PTBS} \em
Let $\Pi$ be a Datalog$^{\neg s}$ program and $D$ an EDB (possibly with duplicates). Then
$\nit{DTBS}(\Pi,D) = \nit{PTBS}(\Pi,D)$ holds.
\end{proposition}

\subsection{Extension of the PBB Semantics to Datalog$^{\pm,\neg sg}$}
\label{sec:dpm-pbb}

We now extend also the PBB semantics to \dpmsn programs $\Pi$.
First, in all predicates of $\Pi$, we add position 0 to carry tid-variables.
Then every rule $\rho \in \Pi$ of the form
$\rho\!: \ \exists \bar{y} H(\bar{y}, \bar{x}) \ \leftarrow \ A_1(\bar{x}_1),A_2(\bar{x}_2),\ldots,A_k(\bar{x}_k)$, with $k>0, \ \bar{x} \subseteq \cup_i \bar{x}_i$, becomes the rule \
$\rho^+\!: \ \exists z \exists \bar{y} \ H(z; \bar{y},\bar{x}) \leftarrow A_1(z_1;\bar{x}_1),A_2(z_2;\bar{x}_2),\ldots,A_k(z_k;\bar{x}_k)$, with fresh, different variables $z, z_1, \ldots,z_k$.
Now consider a rule $\rho \in \Pi$ of the form
$\rho\!: \  \exists \bar{y} H(\bar{y} \bar{x}) \ \leftarrow \ A_1(\bar{x}_1), \ldots, A_k(\bar{x}_k),
\neg B_1(\bar{x}_{k+1}), \ldots, \neg B_\ell(\bar{x}_{k+\ell}),
$
with $\bar{x}_{k+1}, \ldots, \bar{x}_{k+\ell}  \subseteq \bigcup_{i=1}^k \bar{x}_i$;
we transform it
 into the following two rules:
 \begin{eqnarray}
\rho'\!: \  \exists z \exists \bar{y} H(z;\bar{y}, \bar{x}) &~\leftarrow~& A_1(z_1;\bar{x}_1), \ldots, A_k(z_k;\bar{x}_k),
\neg \nit{Aux}_1^\rho(\bar{x}_{k+1}), \ldots, \neg \nit{Aux}_\ell^\rho(\bar{x}_{k+\ell}),~~~~~\label{eq:d++}\\
\nit{Aux}_i^\rho(\bar{x}_{k +i}) &~\leftarrow~& B_i(z_i;\bar{x}_{k+i}), \ \ i = 1, \ldots, \ell.\nonumber
\end{eqnarray}
Analogously to
Theorem \ref{theorem:piplus-is-warded}, the resulting program also belongs to \dpmsn.
Finally, as in Section~\ref{sec:datalog},
the ground atoms in the multiset EDB $D$ are extended in $D^+$ by nulls from $\calI = \{\iota_1, \iota_2, \dots\}$
as tids in position 0
to keep apart duplicates.
For an instance $I$, we write $I_\downarrow$ to denote the set of atoms in $I$ such that
all positions except for position 0 (the tid) carry a ground term (from $\mathbf{C}$).
Analogously to Definition~\ref{def:sem1},
we then define the {\em program-based bag semantics} (PBB semantics) of
a \dpmsn program $\Pi$ and multiset EDB $D$ as
%
$\nit{PBBS}(\Pi,D) := \mc{DI}(\nit{Chase}(\Pi^+,D^+)_\downarrow) =
\{\!\{P(\bar{a})~|~ P(\bar{a})$ is ground and
$P(t;\bar{a}) \in \nit{Chase}(\Pi^+,D^+)\}\!\}$. For a ground atom $A$ (here, without nulls or tids), $\nit{mult}(A,\nit{PBBS}(\Pi,D))$ denotes the multiplicity of $A$ in the multiset $\nit{PBBS}(\Pi,D)$.

\begin{example} \ \label{ex:d+}
(ex.\ \ref{ex:proof-tree} and \ref{ex:ptb-semantics}
cont.)
The transformations of \dpmsn program $\Pi$ and
multiset
EDB $D$ from
Example \ref{ex:proof-tree} are
$D^+ = \{ P(\iota_1, a,b), P(\iota_2, b,c), P(\iota_3, a,d), P(\iota_4, d,c)\}$
and
$\Pi^+ = \{\rho^+_1, \rho^+_2, \rho^+_3, \rho^+_4, \rho^+_5\}$
with:
$\rho_1^+ \!: \ (\exists u) Q(u,x,w) \leftarrow R(v_1,r,y), S(v_2,y,z), T(v_3,x,w)$; \
$\rho_2^+ \!: \ (\exists u,z) R(u,x,z) \leftarrow P(v,x,y)$; \
$\rho_3^+ \!: \ (\exists u,z) S(u,x,z)$ $\leftarrow R(v_1,w,x), T(v_2,w,y)$; \
$\rho_4^+ \!: \ (\exists u) T(u,x,y) \leftarrow P(v,x,y)$; \
$\rho_5^+ \!: \ (\exists u) T(u,x,y) \leftarrow P(v_1,x,z),$ $ T(v_2, z,y)$.

Applying program $\Pi^+$ to instance $D^+$ gives, for example, the following chase sequence for deriving the atoms
$T(\iota_5, d,c), T(\iota_7, a,c),T(\iota_8, b,c),Q(\iota_{11}, a,c)$,
which correspond to the ground atoms in the proof tree in
Figure \ref{fig:proof-tree}, left:
$D^+ \leadsto_{\rho_4^+} T(\iota_5, d,c)
\leadsto_{\rho_2^+} R(\iota_6, a,u)
\leadsto_{\rho_5^+} T(\iota_7, a,c)
\leadsto_{\rho_4^+} T(\iota_8, b,c)
\leadsto_{\rho_3^+} S(\iota_9, u,v)
\leadsto_{\rho_5^+} T(\iota_{10}, a,c)
\leadsto_{\rho_1^+} Q(\iota_{11}, a,c).
$
Implicitly, we thus also turn the PTs rooted at $N_4$, $N_7$, and  $N_9$ into chase sequences.
In total, we get in ${PBBS}(\Pi,D)$ precisely the atoms and multiplicities as in $PTBS(\Pi,D)$
for Example~\ref{ex:ptb-semantics}.
\boxtheorem
\end{example}

In principle, the above transformation $\Pi^+$ and the PBB semantics are applicable to any program $\Pi$ in \dpmsn.
 However, in the first place, we are interested in {\em warded} programs $\Pi$. It is easy to verify that
wardedness of $\Pi$ carries over to $\Pi^+$:

\begin{proposition}
\label{prop:preserve-wardedness}
\em If $\Pi$ is a warded \dpmsn  program, then so is $\Pi^+$.
\end{proposition}

\begin{proof}[Proof Idea] The key observation is that the only additional affected positions in
$\Pi^+$ are the tids at position 0. However, the variables at these positions occur only once in each rule and
are never propagated from the body to the head. Hence, they do not destroy wardedness.
\end{proof}

We conclude this section with the analogous result of Theorem
\ref{thm:isom}:

\begin{theorem}
\label{thm:isom-warded}
\em
For ground atoms $A$, $\nit{mult}(A,\nit{PTBS}(\Pi,D)) = \nit{mult}(A,\nit{PBBS}(\Pi,D))$.
Hence, for a warded \dpmsn program $\Pi$ and multiset EDB $D$, \ $\nit{PTBS}(\Pi,D) = \nit{PBBS}(\Pi,D)$.
\end{theorem}

\section{Decidability and Complexity of Multiplicity}
\label{sec:decidability}

For Datalog$^{\neg s}$ programs, the following problems related to duplicates
have been investigated:
\begin{itemize}
\item FFE: \ Given a program $\Pi$, a database $D$, and a predicate $P$, decide if  every derivable ground  $P$-atom has
 a finite number of DTs.
\item FEE: \ Given a program $\Pi$ and a predicate $P$, decide if every derivable ground $P$-atom has  a finite number of DTs {\em for every} database $D$.
\end{itemize}
\ignore{\begin{eqnarray*}
\nit{FFE} &:=& \{\langle \Pi, P, D\rangle \ | \ \mbox{ every derivable ground  $P$-atom has}
\mbox{ a finite number of DTs}\}. \\
\nit{FEE} &:=& \{\langle \Pi, P\rangle \ | \ \mbox{  every derivable ground $P$-atom has  a finite number} \mbox{ of DTs for every} \ D\}.
\end{eqnarray*} }

\noindent
It has been shown in \cite{DBLP:conf/adc/MumickS93} that FFE is decidable (even
in PTIME data complexity), whereas  FEE is undecidable.
We  extend this study by considering
warded \dpmsn instead of Datalog$^{\neg s}$, and by computing the concrete multiplicities in case of
finiteness. We thus study the following problems:
\begin{itemize}
\item FINITENESS: For fixed warded \dpmsn program $\Pi$: \ Given a multiset database $D$ and a ground atom $A$, does
$A$ have finite multiplicity, i.e., is $\nit{mult}(A,\nit{PBBS}(\Pi,D))$ finite?
\item MULTIPLICITY: For fixed  warded \dpmsn program $\Pi$: \ Given a multiset database $D$ and
a ground atom $A$,
compute the multiplicity of $A$, i.e. $\nit{mult}(A,\nit{PBBS}(\Pi,D))$.
\end{itemize}

\noindent
We will show that both problems defined above can be solved in polynomial time (data complexity).
In case of the
FINITENESS problem, we thus generalize the result of \cite{DBLP:conf/adc/MumickS93}
for the FFE problem from Datalog$^{\neg s}$ to \dpmsn.
In case of the MULTIPLICITY problem, no analogous result for Datalog or Datalog$^{\neg s}$ has existed before.
The following example illustrates that, even if $\Pi$ is a Datalog program with a single rule,
we may have exponential multiplicities. Hence, simply computing all DTs or (in case of \dpm) all PTs is not a viable option if we aim
at polynomial time complexity.

\begin{example}
\label{ex:exponential-multiplicity}
Let $D = \{E(a_0,a_1), E_(a_1,a_2),  \dots, E(a_{n-1},a_n)\} \cup \{ P(a_0,a_1), C(b_0), C(b_1)\}$ for $n \geq 1$
and let $\Pi = \{ \rho\}$ with
$\rho \!:  P(x,y)  \leftarrow P(x,z), E(z,y), C(w)$. Intuitively, $E$ can be considered as an
edge relation and $P$ is the corresponding path relation for paths starting at $a_0$. The $C$-atom in the rule
body (together with the two $C$-atoms in $D$) has the effect that there are always 2 possible derivations to
extend a path. It can be verified by induction over $i$, that atom $P(a_0,a_i)$ has $2^{i-1}$ possible
derivations from $D$ via $\Pi$. In particular,  $P(a_0,a_n)$ has multiplicity $2^{n-1}$.

Note that, if we add atom $E(a_1,a_1)$ to $D$  (i.e., a self-loop, so to speak), then every atom $P(a_0,a_i)$
with $i \geq 1$ has infinite multiplicity. Intuitively, the infinitely many different derivation trees correspond to the arbitrary number of
cycles through the self-loop $E(a_1,a_1)$ for a path from $a_0$ to $a_i$.
\boxtheorem
\end{example}

Our PTIME-membership results will be obtained by appropriately adapting the tractability proof of CQA
for warded \dpm in \cite{DBLP:conf/pods/ArenasGP18}, which is based on the
algorithm \pt  for deciding
if $D \cup \Pi \models P(\bar{c})$ holds for database $D$, warded \dpm program $\Pi$, and ground atom $P(\bar{c})$.
That algorithm works in ALOGSPACE (data complexity), i.e., alternating logspace,
which coincides with PTIME
\cite{DBLP:journals/jacm/ChandraKS81}. It
assumes $\Pi$ to be normalized in such a way that each rule in $\Pi$ is either {\em head-grounded\/} (i.e., each term in the head is a constant
or a harmless variable) or {\em semi-body-grounded\/} (i.e., there exists at most one body atom
with harmful variables).
Algorithm $\pt$ starts with ground atom $P(\bar{c})$ and applies resolution steps until the database $D$ is reached. It thus proceeds as follows:
\begin{itemize}
\item If $P(\bar{c}) \in D$, then accept. Otherwise, guess a head-grounded rule $\rho \in \Pi$
whose head can be matched to $P(\bar{c})$ (denoted as $\rho \rhd P(\bar{c})$). Guess an instantiation $\gamma$
on the variables in the body of $\rho$ so that $\gamma(\head(\rho)) = P(\bar{c})$.
Let $\calS = \gamma(\body(\rho))$.
\item Partition $\calS$ into $\{\calS_1, \dots, \calS_n\}$, such that each null occurring in $\calS$
occurs in exactly one $\calS_i$, and each set $\calS_i$ is chosen subset-minimal with this property.
The purpose of these
sets $\calS_i$ of atoms is to keep together, in the parallel universal computations of
\pt, the nulls in $\calS$ until the atom in which they are created is known.
\item Universally select each set $\calS' \in \{\calS_1, \dots, \calS_n\}$ and ``prove'' it:
If $\calS'$ consists of a single ground atom $P'(\bar{c}')$, then call \pt recursively for $D$, $\Pi$, $P'(\bar{c}')$.
Otherwise, do the following:

(1) For each atom $A \in \calS'$, guess rule $\rho_A$ with $\rho_{A} \rhd A$ and guess
variable instantiation $\gamma_{A}$ on the variables in the body of
$\rho_{A}$ such that $A = \gamma_{A}(\head(\rho_A))$.

(2) The set $\bigcup_{A \in \calS'} \gamma_{A}(\body(\rho_A))$ is
partitioned as above and each component of this partition is proved in a parallel universal
computation.
\end{itemize}

The key to the ALOGSPACE complexity of algorithm $\pt$ is that the data structure
propagated by this algorithm fits into logarithmic space. This data structure is given by
a pair $(\calS,\calRS)$, where
$\calS$ is a set of atoms (such that $|\calS|$ is bounded by the maximum number of body atoms of
the rules in $\Pi$) and  $\calRS$ is a set of pairs $(z,x)$, where $z$ is a null occurring in $\calS$ and
$x$ is either an atom $A$ (meaning that null $z$ was created when the application of some rule $\rho$
generated the atom $A$ containing this null $z$) or
the symbol $\varepsilon$ (meaning that we have not yet found such an atom $A$).
A {\em witness\/}
for a successful computation of \pt is then given by a tree with existential and universal nodes, with an
existential node $N$ consisting of a pair $(\calS,\calRS)$;
a universal node $N$ indicates the guessed rule $\rho_{A}$ together with the
instantiation $\gamma_{A}$ for each atom $A \in \calS$ at the (existential) parent node
of $N$. The child nodes of each universal node $N$ are obtained by
partitioning $\bigcup_{A \in \calS} \gamma_{A}(\body(A))$ as
described above and computing the corresponding set $\calRS$.
At the root, we thus have an existential node labelled $(\calS,\calRS) = (\{P(\bar{c})\}, \emptyset)$.
Each leaf node is a universal node labelled with $(B, \emptyset)$ for some
ground atom $B \in D$. Hence, such a node corresponds to an accept-state.

\begin{example}
\label{ex:witness}
Recall $\Pi$ and $D$ from Example~\ref{ex:proof-tree}. A proof tree of
ground atom $Q(a,c)$ is shown in Figure \ref{fig:proof-tree} on the left.
On the right, we display the witness of the corresponding successful computation of \pt.
Note that witnesses have a strict alternation of existential and universal nodes.
In the witness in Figure \ref{fig:proof-tree}, we have merged each existential node and its
unique universal child node to make the correspondence between a PT and a witness yet more visible.
In particular,
on each depth level of the trees, we have exactly the same set of atoms (with the only difference, that in the witness
these atoms are grouped together to sets $\calS$ by null-occurrences). In this simple example, only node $M_{23}$
contains such a group of atoms, namely the atoms from the nodes $N_2$ and $N_3$ in the PT.

In order not to overburden the figure, we have left out the sets $\calRS$ carrying the information on the
introduction of nulls. In this simple example, the only node with a non-empty set $\calRS$ is $M_6$, with
$\calRS = \{(u,R(a,u)\}$, i.e. we have to pass on the information that the null $u$ in node $M_{23}$ was
introduced by applying some rule (namely $\rho_2$) which generated the atom $R(a,u)$. We thus ensure
(when proceeding from node $M_6$ to node $M_{10}$) that null $u$ is introduced in the same way as before.

The information from the universal node below each existential node is displayed in the second line of each node: here we have pairs consisting of the guessed rule $\rho$ plus the guessed instantiation $\gamma$ for each atom in the corresponding set $\calS$ (as mentioned above, $\calS$  is a singleton in all nodes except for $M_{23}$).
For instance, $\gamma_1$ in node $M_{1}$ denotes the substitution
$\gamma_1 = \{x \leftarrow a, y \leftarrow u, z \leftarrow v, w \leftarrow c\}$.
Only in node $M_{23}$, we have 2 pairs $(\rho_2,\gamma_2), (\rho_3,\gamma_3)$ with
$\gamma_2 = \{x \leftarrow a, y \leftarrow b\}$ and
$\gamma_3 = \{w \leftarrow a, x \leftarrow u, y \leftarrow c\}$. Note
that the subscripts $i$ of the $\gamma_i$'s in Figure \ref{fig:proof-tree}
are to be understood local to the node. For instance,
the two different applications of rule $\rho_4$ in nodes $M_9$ and $M_{12}$ are with
the instantiations $\gamma_4 = \{x \leftarrow b, y \leftarrow c\}$ (in $M_9$) and
$\gamma_4 = \{x \leftarrow d, y \leftarrow c\}$ (in $M_{12}$), respectively.
\boxtheorem
\end{example}

The above example illustrates the close relationship between PTs $\calT$ and witnesses $\calW$
of the \pt algorithm.
However, our goal is a one-to-one correspondence, which requires further measures:
for witness trees, we thus assume from now on (as in Example \ref{ex:witness}) that
each existential node is merged with its unique universal child node and that nulls are renamed apart:
that is, whenever nulls are introduced by a resolution step of \pt,
then all nulls in
$\nulls(\gamma(\body(\rho))) \setminus \nulls(\gamma(\head(\rho)))$ must be fresh (that is, they must not occur elsewhere
in $\calW$). Moreover, we eliminate redundant information from PTs and witnesses by pruning repeated subtrees:
let $N$ be a node in a PT $\calT$ such that some null is introduced via atom $A$ in $\lambda(N)$, and let $N$ be closest to the root with this property, then prune all other subtrees below all nodes $M \neq N$ with $\lambda(M) = A$.
Likewise, if a node in a witness $\calW$
contains an atom $A$ and the pair $(z,A)$ in $\calRS$, then we omit the resolution step for $A$.
We refer to the reduced PT of $\calT$ as $\calT^*$ and to the reduced witness of $\calW$ as $\calW^*$.
For instance, in the PT in Figure~\ref{fig:proof-tree}, we delete the node $N_{10}$ because the information on how to derive
atom $R(a,u)$ is already contained in the subtree rooted at node $N_2$.
Likewise, in the witness in Figure~\ref{fig:proof-tree},
we omit the resolution step for atom $R(a,u)$ in node $M_6$, because, in this node, we have $\calRS = \{(u,R(a,u)\}$. Hence,
it is known from some resolution step ``above'' that $u$ must be introduced via this atom and its derivation is
checked elsewhere. We thus delete $M_{10}$.

It is straightforward to construct a reduced witness $\calW^*$ from
a given reduced PT $\calT^*$, and vice versa. We refer to these constructions as $\TtoW$ and $\WtoT$, respectively:
Given $\calT^*$, we obtain $\calW^*$ by a top-down traversal of $\calT^*$ and merging siblings if they share a null.
Moreover, if the label at a child node in $\calT^*$ was created by applying rule $\rho$ with substitution $\gamma$, then
we add $(\rho,\gamma)$ to the node label in $\calW^*$. Finally, if such a rule application generates a new null $z$, then we add $(z,\gamma(\head(\rho))$ to $\calRS$ of every child node which still contains null $z$.
Conversely, we obtain $\calT^*$ from a reduced witness $\calW^*$
by a bottom-up traversal of $\calW^*$, where we turn every node labelled with $k$ atoms into $k$ siblings, each labelled with one of these atoms. The edge labels $\rho$ in $\calT^*$ are obtained from
the corresponding $(\rho, \gamma)$ labels in the node in $\calW^*$.  In
summary, we have:

\begin{lemma} \label{lemma:corresp} \em
There is a one-to-one correspondence \ignore{relationship} between reduced proof trees $\calT^*$ for a ground atom $P(\bar{c})$
and reduced witnesses $\calW^*$ for its successful \pt computations. More precisely, given $\calT^*$,
we get a reduced witness as $\TtoW(\calT^*)$;
given $\calW^*$, we get a reduced PT as $\WtoT(\calW^*)$
with  $\WtoT(\TtoW(\calT^*)) = \calT^*$ and
$\TtoW(\WtoT(\calW^*)) = \calW^*$. \boxtheorem
\end{lemma}

So far, we have only considered warded \dpm, without negation.
However, this restriction is inessential. Indeed,
by the definition of warded \dpmsn,
negated atoms can never contain a null in the chase. Hence, one can easily get
rid of negation (in polynomial time data complexity) by computing for  one stratum after the other the answer $Q(\Pi,D)$
to query $Q$ with $Q \equiv P(x_1, \dots, x_n)$, i.e., all ground atoms $P(a_1, \dots, a_n)$ with
$\Pi \cup D \models P(a_1, \dots, a_n)$.
We can then replace all occurrences of $\neg P(t_1, \dots, t_n)$ in any rule of
$\Pi$ by the positive atom $P'(t_1, \dots, t_n)$ (for a new predicate symbol $P'$)
and add to the instance all ground atoms $P'(a_1, \dots, a_n)$
with $(a_1, \dots, a_n) \not\in Q(D, \Pi)$.
Hence, all our results proved in this section for warded \dpm also hold for warded \dpmsn.

Clearly, there is a one-to-one correspondence between proof trees PT $\calT$ and their reduced forms $\calT^*$. Hence,
together with Lemma \ref{lemma:corresp}, we can compute the multiplicities
by computing (reduced) witness trees.
This allows us to obtain the results below:

\begin{theorem}
\label{theo:infinite-multiplicity} \em
Let $\Pi$ be a warded \dpmsn program and
$D$ a database (possibly with duplicates).
Then there exists a bound $K$ which is polynomial in $D$, s.t.\ for every
ground atom~$A$:
\begin{enumerate}
\item $A$ has finite multiplicity if and only if all reduced witness trees of $A$
have depth $\leq K$.
\item If $A$ has infinite multiplicity, then there exists at least one reduced witness tree of $A$ whose depth is in $[K+1, 2K]$.
\end{enumerate}
\end{theorem}


\begin{proof}[Proof Idea] Recall that the data structure propagated by the \pt algorithm consists of pairs $(\calS, \calRS)$.
We call pairs $(\calS,\calRS)$ and $(\calS',\calR_{\calS'})$
{\em equivalent\/} if one can be obtained from the other by renaming of nulls.
The bound $K$ corresponds to the maximum number of non-equivalent pairs $(\calS,\calRS)$ over the given signature and domain of $D$.
By the logspace bound on this data structure, there can only be polynomially many (w.r.t.\ $D$) such pairs.
For the first claim of the theorem, suppose that a (reduced) witness tree $\calW^*$ has depth greater than $K$;
then there must be a branch with two nodes
$N$ and $N'$ with equivalent pairs $(\calS,\calRS)$ and $(\calS',\calR_{\calS'})$. We get infinitely many witness trees by
arbitrarily often
iterating the path
between $N$ and $N'$.
\end{proof}

In principle, Theorem \ref{theo:infinite-multiplicity} suffices to prove decidability of
FINITENESS and design an algorithm for the
MULTIPLICITY: just chase database $D$ with the transformed warded Datalog program $\Pi^+$
up to
depth $2K$. If the desired ground atom $A$ extended by some tid is generated at
a depth greater than $K$, then conclude that $A$ has infinite multiplicity. Otherwise,
the multiplicity of $A$ is equal to the number of atoms of the form $(\iota;A)$
in the chase result. However, this chase of depth $2K$ may produce an exponential number of atoms and hence take exponential time. Below we show that we can in fact do significantly better:

\begin{theorem}
\label{theo:finite-multiplicity} \em
For warded \dpmsn programs, both the FINITENESS problem and the MULTIPLICITY problem
can be solved in polynomial time.
\end{theorem}

\begin{proof}[Proof Idea]
A decision procedure for the FINITENESS problem can be obtained by modifying the ALOGSPACE algorithm  \pt
from \cite{DBLP:conf/pods/ArenasGP18} in such a way that we additionally ``guess'' a branch in the witness tree
with equivalent labels. The additional information thus needed also fits into logspace.

The MULTIPLICITY problem can be solved in polynomial time by a tabling approach to the \pt algorithm. We thus
store for each (non-equivalent) value of $(S,R_S)$ how many (reduced) witness trees it has and propagate this information
upwards for each resolution step encoded in the (reduced) witness tree.
\end{proof}

\section{Multiset Relational Algebra (MRA)}
\label{sec:next}

Following \cite{DBLP:conf/slp/MaherR89,DBLP:conf/vldb/MumickPR90}, we consider multisets (or bags) $M$ and  elements $e$ (from some domain) with non-negative integer multiplicities,
$\mult{e}{M}$ (recall from Section \ref{sec:dtbs} that, by definition,
$e \in M$ iff $\mult{e}{M} \geq 1$).
Now consider multiset relations $R, S$. Unless stated otherwise, we
assume that $R,S$ contain tuples of the same arity, say $n$.
%
%
We define the following multiset operations of MRA: the  {\em multiset union}, $\uplus$, is defined
by $R \uplus S := T$, with $\mult{e}{T} := \mult{e}{R} + \mult{e}{S}$.
{\em Multiset selection}, $\sigma^\nit{m}_C(R)$, with a condition $C$, is defined as the multiset $T$
containing all tuples in $R$ that satisfy $C$ with  the same multiplicities as  in $R$.
For {\em multiset projection} $\pi^\nit{m}_{\bar{k}}(R)$, we get the multiplicities
$\mult{e}{\pi^\nit{m}_{\bar{k}}(R)}$ by
summing up the multiplicities of all tuples in $R$
that, when projected to the positions $\bar{k} = \langle i_1, \ldots, i_k\rangle$, produce $e$.
For the {\em multiset (natural) join} ${R \Join^{\!\nit{m}}\! S}$, the multiplicity of each tuple $t$ is
obtained as the product of multiplicities of tuples from $R$ and of tuples from $S$ that join to $t$.

For the {\em multiset difference},
two definitions are conceivable: \ Majority-based, or ``monus" difference
(see e.g.\ \cite{DBLP:journals/japll/GeertsP10}),
given by \ $R \ominus S := T$,
\ignore{$R \smallsetminus_\nit{mb} S := T$,}  with  $\mult{e}{T} := \nit{max}\{\mult{e}{R}-\mult{e}{S},0\}$.
There is also the ``all-or-nothing" difference:  \ $R \otimes S := T$,
\ignore{$R \smallsetminus_{\it an} S := T$,}  with  $\mult{e}{T} := \mult{e}{R} \mbox{ if }$ $ e \notin S$ and
$\mult{e}{T} := 0$, otherwise\ignore{ \
(``an" stands for  {\em all-or-nothing\/})}. Following \cite{DBLP:conf/adc/MumickS93},
we have only considered $\otimes$  so far (implicitly, starting with Datalog$^{\neg s}$, in Section \ref{sec:dtbs}).
The {\em multiset intersection} $\cap_m$  is not treated or used in
any of \cite{DBLP:conf/semweb/AnglesG16,DBLP:conf/slp/MaherR89,DBLP:conf/vldb/MumickPR90,DBLP:conf/adc/MumickS93}.  Extending the DTB semantics from \cite{DBLP:conf/vldb/MumickPR90}
to
 $\cap_m$ would treat it as a special case of
the join, which may be counter-intuitive, e.g.,
$\{a,a,b\} \cap_\nit{m} \{a,a,a,c\} := \{a,a,a,a,a,a\}$.
Alternatively, we could define ``minority-based" intersection $R \cap_\nit{mb} S$,
which returns each element with its minimum  multiplicity, e.g.,
$\{a,a,b\} \cap_\nit{mb} \{a,a,a,c\} := \{a,a\}$.

We consider MRA with  the following basic multiset
operations: multiset union $\uplus$\ignore{$\cup_m$}, multiset projection $\pi^\nit{m}_{\bar{k}}$,
multiset selection $\sigma^\nit{m}_C$, with $C$ a condition,
multiset (natural) join $\Join^{\!\nit{m}}$, and
(all-or-nothing) multiset difference
\ignore{$\smallsetminus_{\it an}$}$\otimes$. We can also include {\em duplicate elimination} in MRA, which becomes operation $\mc{DI}$ of Definition \ref{def:di-and-sp} when using tids. Being just a projection in the latter case, it can be represented in Datalog. For the moment, we consider multiset-intersection as a special case of
multiset (natural) join $\Join^{\!\nit{m}}$.
It is well known that the basic, set-oriented relational algebra operations can all be captured
by means of (non-recursive) Datalog$^{\neg s}$ programs (cf.\ \cite{AHV94}).
Likewise, one can capture MRA by means of (non-recursive) Datalog$^{\neg s}$ programs with multiset semantics
(see e.g.\ \cite{DBLP:conf/semweb/AnglesG16}).
Together with our transformation into set semantics of \dpmsn, we thus obtain:

\begin{theorem}
\label{theo:eMRA}
 \em
The Multiset Relational Algebra (MRA) can be represented by warded \dpmsn with
set semantics.
\boxtheorem
\end{theorem}
As a consequence, the MRA operations can still be performed in polynomial-time (data complexity) via \dpmsn.
This result tells us that MRA -- applied at the level of an EDB with duplicates --  can be
represented in warded Datalog, and uniformly integrated under the same logical semantics with an ontology represented in warded Datalog.

We now retake multiset-intersection (and later also multiset-difference),
which appears as \ $\cap_m$ and $\cap_\nit{mb}$.
The former does not offer any problem for our representation in \dpm \ as above, because it is a special case of multiset join.
\ignore{***************************
However, still in this case, a self-intersection, $R \cap_m R$, treated as a join,  would return the multiplication of multiplicities. For example,
for $R = \{(a,b), (a,b), (c,d)\}$, $R \cap_m R = \{(a,b),(a,b), (a,b), (a,b), (c,d)\}$. \ The same result would be obtained via the intended  \dpm \ rule, namely $\exists z J(z;x,y) \leftarrow R(z_1;x,y),$ $ R(z_2;x,y)$, which
could be undesirable in comparison with the possibly more natural $\{(a,b),(a,b), $ $ (c,d)\}$. Then,
we could use the following rule instead: \ $\exists z J(z;x,y) \leftarrow R(z;x,y), R(z;x,y)$, i.e. with the same tid variable in the rule body, which makes sense  since the same tids appear in the two versions of relation $R$. Since  tids are global, a rule like this works as intended only for the self-intersection (or full self-join), i.e. with the same predicate, but not for arbitrary joins or intersections.
***************************}
In contrast, the  {\em minority-based} intersection, $\cap_\nit{mb}$, is more problematic.\footnote{The {\em majority-based union} operation on bags, that returns, e.g. $\{\{a, a,b,c\}\} \cup_\nit{mab} \{\{a, b,b\}\}:= \{\{a, a,b,b,c\}\}$ should be equally problematic.}\ First, the DTB semantics does not give an account of it in terms of Datalog that we can use to build upon. Secondly, our \dpm-based formulation of duplicate management with MRA operations is set-theoretic. Accordingly, to investigate the representation of the bag-based operation $\cap_\nit{mb}$ by means of the latter, we have to agree on a set-based reformulation $\cap_\nit{mb}$. We propose for it a tid-based (set) representation, because  tid creation becomes crucial to make it a deterministic operation.  Accordingly, for  multi-relations $P$ and $R$ with the same arity $n$ (plus $1$ for tids), we define:
\begin{equation}P \cap_\nit{mb} Q := \{(i;\bar{a}) \ | \ (i;\bar{a}) \in \left\{\begin{array}{ll}
                                                                            P & \mbox{ if } |\pi_0(\sigma_{1..n =\bar{a}}(P))| \leq |\pi_0(\sigma_{1..n =\bar{a}}(Q))|\\
                                                                            Q & \mbox{ otherwise}\end{array} \right. \}. \label{eq:sqlinter}\end{equation}
Here, we only assume that tids are {\em local} to a predicate, i.e. they act as values for a surrogate key. Intuitively, we keep for each tuple in the result the duplicates that appear in the relation that contains the minimum number of them. Here, $\pi_0$ denotes the projection on the $0$-th attribute (for tids), and $\sigma_{1..n =\bar{a}}$ is the selection of those tuples which
coincide with $\bar{a}$  on the next $n$ attributes.
This operation may be non-commutative when equality holds in the first case of (\ref{eq:sqlinter}) (e.g. $\{\langle 1; a\rangle\} \cap_\nit{mb} \{\langle 2; a\rangle\} = \{\langle 1; a\rangle\} \neq \{\langle 2; a\rangle\} = \{\langle 2; a\rangle\} \cap_\nit{mb} \{\langle 1; a\rangle\}$). Most importantly,
it is non-monotonic: if any of the extensions of $P$ or $Q$ grows, the result may not contain the previous result,\footnote{We could redefine (\ref{eq:sqlinter}) by introducing new tids, i.e. tuples $(f(i),\bar{a})$, for each tuple ($i,\bar{a})$ in the condition in (\ref{eq:sqlinter}), with  some function $f$ of tids. The $f(i)$ could be the next tid values after the last one used so far in a list of them. The  operation defined in this would still be non-monotonic.} e.g. $\{\langle 1; a\rangle\} \cap_\nit{mb} \{\langle 2; a\rangle\} = \{\langle 1; a\rangle\} \not \supseteq \{\langle 2; a\rangle\} = (\{\langle 1; a\rangle\} \cup \{\langle 3; a\rangle\}) \cap_\nit{mb} \{\langle 2; a\rangle\}$. \
(It is still non-monotonic under the DTB semantics.)
We get the following inexpressibility results:

\begin{proposition}
\label{prop:FOinter} \em
The minority-based intersection with duplicates $\cap_\nit{mb}$ as in (\ref{eq:sqlinter})
cannot be represented in Datalog,
(positive) \dpm, or FO predicate logic (FOL). The same applies to the majority-based (monus) difference, $\ominus$.
\end{proposition}

 \begin{proof}[Proof Idea]
By the non-monotonicity of $\cap_\nit{mb}$, it
is clear for Datalog and (positive) \dpm. For the inexpressibility in FOL, the key idea is that with the help
of $\cap_\nit{mb}$ or $\ominus$ we could express the  {\em majority quantifier},
which is known to be undefinable in FOL \cite{vb,wester}.
 \end{proof}

\begin{proposition} \em
\label{prop:stdatinter}
The minority-based intersection with duplicates $\cap_\nit{mb}$  as in (\ref{eq:sqlinter})
cannot be represented in Datalog$^{\neg s}$.
The same applies to the majority-based difference, $\ominus$.
\end{proposition}

\begin{proof}[Proof Idea]
It can be shown that if a logic is powerful enough to express
any of $\cap_\nit{mb}$ or $\ominus$, then we could express in this logic -- for sets $A$ and $B$ --
that $|A| = |B\smallsetminus A|$ holds. However, the latter property cannot even be expressed
in the logic $L^\omega_{\infty\omega}$ under finite structures \cite[sec. 8.4.2]{ef} and this logic extends
Datalog$^{\neg s}$.
\end{proof}

Among future work, we plan to investigate further inexpressibility issues such as, for instance, whether \dpmsn is expressive enough to capture
$\cap_\nit{mb}$  and $\ominus$.
More generally, the development of tools to address (in)expressibility results in \dpm,  with or without negation, is a matter of future research.

\section{Conclusions and Future Work}
\label{sec:conclusion}

We have proposed the specification of the bag semantics of Datalog in terms of
warded Datalog$^{\pm}$ with set semantics and
we have extended this specification to Datalog$^{\neg s}$ as well as warded \dpm and \dpmsn.
That is, the bag semantics of
all these languages can be captured by warded \dpmsn with set semantics.
We have also discussed Multiset Relational Algebra (MRA)
as an immediate application of our~results.

Our work underlines that warded Datalog$^{\pm}$ is indeed a well-chosen fragment of Datalog$^{\pm}$:
it provides a mild extension of Datalog by the restricted use of existentially quantified variables in the rule heads,
which suffices to capture certain forms of ontological reasoning \cite{DBLP:conf/pods/ArenasGP14,DBLP:conf/ijcai/GottlobP15}
and, as we have seen here, the bag semantics of Datalog. At the same time, it maintains favorable properties of Datalog, such as polynomial-time query answering. Actually, the techniques developed for establishing this polynomial-time complexity result
for warded Datalog$^{\pm}$
in \cite{DBLP:conf/pods/ArenasGP14,DBLP:conf/pods/ArenasGP18} have also greatly helped us to prove our
polynomial-time results for the FINITENESS and MULTIPLICTY problems.

Another advantage of warded Datalog$^{\pm}$
is the existence of an efficient implementation in the VADALOG system \cite{DBLP:conf/ijcai/BellomariniGPS17}.
Further extensions of this system -- above all the support of SPARQL with bag semantics based on
the Datalog re\-writing proposed in \cite{DBLP:conf/semweb/AnglesG16,DBLP:journals/jancl/PolleresW13}
-- are currently under way.
%
Recall that warded \dpmsn was shown to capture a core fragment of SPARQL
under  OWL 2 QL entailment regimes \cite{glimm},
with set semantics though \cite{DBLP:conf/pods/ArenasGP14,DBLP:conf/ijcai/GottlobP15}.
Our transformation via warded \dpmsn will allow us to capture also the bag semantics.

We are currently also working on further inexpressibility results.
Recall that our translation into \dpmsn does not cover multiset intersection;
moreover,
multiset difference is only handled in the all-or-nothing form
$\otimes$,
while the
sometimes more natural form \blue{$\ominus$} has been left out (but the former is good enough for  application to SPARQL).
We conjecture that the two operations $\cap_\nit{mb}$ and $\ominus$
are not expressible in \dpmsn with set semantics.
The verification of this conjecture is a matter of ongoing work.

Finally, we plan to extend our treatment \dpmsn with bag semantics. Recall that  in Section~\ref{sec:warded},
we have only defined multiplicities for ground atoms (without nulls). With our methods developed here,
we can also compute for a non-ground atom $p(u,a)$ (more precisely, with $u$ a null in $\mathbf{B}$) how often it is derived from $D$ via $\Pi$.
\ignore{just add the  rule $q(y) \leftarrow p(x,y)$ for fresh predicate $q$ and find the multiplicity of $q(a)$.}%
However, answering questions like ``for how many different null values $u$ can $p(u,a)$ be derived?'' requires new methods.


%


\bibliography{bagsemantics}

\clearpage
\appendix

\section{Proofs for Section 3}
\label{app:Section3}

\begin{proof}[Proof of Theorem \ref{theorem:piplus-is-warded}]
It is easy to verify that, in every $\rho^+ \in \Pi^+$, the only harmful position is the
position 0 (i.e, the tid) of each predicate from $\Pi$. However, the variables occurring in position 0
in the rule bodies do not occur in the head. Hence, none of the rules in $\Pi^+$ contains a dangerous variable and,
therefore, $\Pi^+$ is trivially warded. By the same consideration, the auxiliary predicates introduced in
(\ref{eq:aftRew}) above contain only non-affected positions. Since these are the only negated atoms in rules of $\Pi^+$,
we have only ground negation in $\Pi^+$. Finally, the stratification of $\Pi$ carries over to $\Pi^+$, where
an auxiliary predicate
$\nit{Aux}_i^\rho$ introduced for a negated predicate $B_i$ in the definition of predicate  $H$ by $\rho$ in $\Pi$ ends up in stratum $s$ of $\Pi^+$ iff $B_i$ is
in stratum $s$ in $\Pi$ (cf.\ (\ref{eq:forRew}) and (\ref{eq:aftRew})).
\end{proof}

\begin{proof}[Proof of Theorem \ref{thm:isom}]
We first consider the case of a Datalog program $\Pi$ (i.e., without negation).
It suffices to prove the following claim: for a Datalog program $\Pi$ with  multiset EDB $D$,  there is a one-to-one correspondence between DTs in $\nit{DT}(\Pi,\nit{col}(D))$ with root atoms \ignore{of the form} $P(\bar{c})$  and minimal chase-sequences with $\Pi^+$ that start with $D^+$ and end in atoms  $P(\iota;\bar{c})$, with $\iota \in \mc{I}$, i.e. that establish $D^+ \leadsto^* P(\iota;\bar{c})$. By ``minimality'' of a chase-sequence we mean that every intermediate atom derived
is used to enforce some rule later along the same sequence. \
More precisely:
\noindent (a) \  $\bl{\delta\!: \ \nit{col}^{-1}(\nit{DTBS}(\Pi,D)) \cong \nit{col}^{-1}(\nit{Chase}(\Pi^+,D^+))}$, i.e. there are isomorphic as sets.
\noindent (b) \ For every element $e$ (in the data domain): \
$\nit{mult}(e,\nit{col}^{-1}(\nit{DTBS}(\Pi,D))) = \nit{mult}(e,$ $\nit{Chase}(\Pi^+,D^+))$.

One direction is by induction on the depth of the trees (the other direction is similar, and by induction of the  length of the chase sequences). The correspondence is clear for depth 1. Now  assume that, for every DT of depth at most $n$, with some root $Q(\bar{c})$, there is exactly one chase sequence  ending in  $Q(\iota;\bar{c})$.

We now consider trees  of depth $n+1$. Let $P(\bar{c})$ be an atom at the root of a DT $\calT$ obtained from the ground instantiation $\rho_g\!: \ P(\bar{c}) \leftarrow  Q_1(\bar{c}_1), \ldots,$ $ Q_k(\bar{c}_k)$ of a rule $\rho\!: \ P(\bar{x}) \leftarrow Q_1(\bar{x}_1), \ldots, Q_k(\bar{x}_k)$. Then it has distinct children $Q_1(\bar{c}_1), \ldots, Q_k(\bar{c}_k)$, in this order from left ro right, each of them with a  DT $\calT_1, \ldots, \calT_k$ of depth at most $n$. These DTs in their turn correspond to distinct chase sequences, $s_1, \ldots, s_k$,  starting in $D^+$ and with end atom of the form $Q_1(\iota_1';\bar{c}_1), \ldots, Q_1(\iota_k';\bar{c}_k)$, with different $\iota_j'$s. Then, the interleaved combination of these chase-sequences -- to respect the canonical order of rule applications and body atoms -- and concatenated suffix ``$ \leadsto_{\rho_g^+} \! P(\iota;\bar{c})$" \ is  exactly the one deriving $P(\iota;\bar{c})$.

It remains to consider the case of a Datalog$^{\neg s}$ program $\Pi$.
In this case, we have to consider a chase procedure for Datalog$^{\pm, \neg s}$ applied to $\Pi^+$, and its correspondence to DTs for
$\Pi$.
The chase for Datalog$^{\pm, \neg s}$ \cite{DBLP:journals/ws/CaliGL12} is similar to that for \dpm, except that now it is defined in a  stratified manner: If, in the course of a chase sequence the positive part $\nit{body}^+$ of
a Datalog$^{\pm, \neg s}$ rule with body $\nit{body}^+, \nit{body}^-$  becomes applicable, the rule is applied as long as the instantiated atoms in $\nit{body}^-$ have not been generated already, at a previous stratum of the chase (see \cite[sec. 10]{DBLP:journals/ws/CaliGL12} for details). This has exactly the same effect as the use of {\em negation-as-failure} with stratified Datalog$^{\neg s}$,
as can be easily proved by induction of the strata.
\end{proof}

\begin{proof}[Proof of Corollary \ref{cor:model}]
We first consider the case of a Datalog program $\Pi$ (i.e., without negation).
Let $M$ be the minimal model of $\Pi \cup  \mc{SP}(D)$. Then, $\mc{SP}(\nit{Chase}(\Pi^+,D^+))$ $ \subseteq M$. This can be seen as follows:  every atom $A$ in the former has at least one chase-derivation with $\Pi^+ \cup D^+$, and then,
as in the proof of Theorem \ref{thm:isom},
a DT with $\Pi \cup \nit{col}(D)$, which is also a derivation tree from $\Pi \cup \cal{SP}(D)$. By the correspondence between top-down and bottom-up evaluations for Datalog,  $A \in M$. The other direction is similar.

For a Datalog$^{\neg s}$ program $\Pi$, we use the correspondence between the stratified chase-sequences with $\Pi^+ \cup D^+$ and the stratified, bottom-up computation of the standard model for $\Pi \cup \mc{SP}(D)$, as shown in the proof of
Theorem \ref{thm:isom}.
 \end{proof}

\section{Proofs for Section 4}
\label{app:Section4}

\begin{proof}[Proof of Proposition \ref{prop:DTBS-vs-PTBS}]
Due to the satisfaction by a Datalog$^{\neg s}$ program of conditions (1)-(4) on proof trees at the beginning of this section (conditions (5)-(8) do not apply to such a program), every PT for an atom $A$ from multiset $D$ and  a Datalog$^{\neg s}$ program $\Pi$ $D$ is also a DT for the atom, and the other way around.
\end{proof}

\begin{proof}[Proof of Proposition \ref{prop:preserve-wardedness}]
Through the transformation of rules  of $\Pi$ into rules of $\Pi^+$,
the non-zero positions in $\Pi^+$ coincide with those in $\Pi$. Hence,
$\nit{Aff}(\Pi) = \nit{Aff}(\Pi^+) \smallsetminus \{P[0]~|~P$ appears in a rule head$\}$
holds.
Then, $\nit{NonAff}(\Pi^+) = \nit{NonAff}(\Pi) \cup \{P[0]~|~P$
does not appear in a rule head$\}$; and,
furthermore,
$\nit{Harmless}(\rho^+) = \nit{Harmless}(\rho) \cup \{z_i~|~ A_i$
does not appear in a rule head$\}$.
Since tid-variables in bodies do not appear in heads, $\nit{Dang}(\rho^+) =\nit{Dang}(\rho)$ holds.
Thus, for each rule $\rho^+$, the variables in $\nit{Dang}(\rho^+)$ already have a  ward in $\rho$.
\end{proof}

\begin{proof}[Proof of Theorem \ref{thm:isom-warded}]
It suffices to prove that there is a one-to-one correspondence between normalized, non-equivalent PTs in $\nit{PT}(\Pi,\nit{col}(D))$ with root atoms  $P(\bar{c})$  and (minimal) chase-sequences with $\Pi^+$ that start with $D^+$ and end in atoms  $P(\iota;\bar{c})$, with $\iota \in \mc{I}$. One direction of the correspondence is is obtained by induction on the depth of PTs (the other one is similar, by induction on the length of chase sequences). The assumptions on canonical orders of application of rules and generation of PTs are crucial.

\ignore{More precisely:
\noindent (a) \  $\bl{\delta\!: \ \nit{col}^{-1}(\nit{PTBS}(\Pi,D)) \cong \nit{col}^{-1}(\nit{Chase}(\Pi^+,D^+))}$, i.e. there are isomorphic as sets.
\noindent (b) \ For every element $e$ (in the data domain): \
$\nit{mult}(e,\nit{col}^{-1}(\nit{PTBS}(\Pi,D))) = \nit{mult}(e,$ $\nit{Chase}(\Pi^+,D^+))$. }

 The correspondence is clear for PT depth 1. Now  assume that, for every PT of depth at most $n$, with some root $Q(\bar{c})$, there is exactly one chase sequence  ending in $Q(\iota;\bar{c})$.
\ We now consider a PT $\calT$ of depth $n+1$ with an atom $P(\bar{c})$ at the root, and obtained from the ground instantiation $\rho_g\!: \ P(\bar{c}) \leftarrow  Q_1(\bar{c}_1), \ldots,$ $ Q_k(\bar{c}_k)$ of a rule $\rho\!: \ \exists z P(\bar{x},z) \leftarrow Q(\bar{x}_1), \ldots, Q(\bar{x}_k)$.\footnote{Variables  do not have to appear  in this order in the head, but we may assume that then original programs have at most a single existential in the head \cite{DBLP:conf/pods/ArenasGP18}.} Then it has distinct children $Q_1(\bar{c}_1), \ldots, Q_k(\bar{c}_k)$, in this order from left to right, each of them with a PT $\calT_1, \ldots, \calT_k$ of depth at most $n$ when the atom is positive, and just a (negative) leaf when the atom is negative. The former in their turn correspond to distinct chase sequences $s_i$ starting in $D^+$ and  ending in  atom $Q_i(\iota_i';\bar{c}_i)$, with different $\iota_j'$s. These sequences can be interleaved in order to respect the canonical order of rule application, and the suffix ``$ \leadsto_{\rho_g^+} \! P(\iota;\bar{c})$" can be added;  resulting in a sequence $s$ for the Datalog$^{\neg s,g}$ program $\Pi^+$  that derives $P(\iota;\bar{c})$ \cite{DBLP:journals/ws/CaliGL12}.
\end{proof}

\section{Proofs for Section 5}
\label{app:Section5}

\begin{proof}[Proof of Theorem~\ref{theo:infinite-multiplicity}]
Recall that the \pt algorithm propagates data structures $(\calS,\calRS)$ consisting of a set $\calS$ of atoms and a set $\calRS$ of pairs $(z,x)$, where $z$ is a null in $\calS$  and $x$ is an atom (i.e., the atom by which null $z$ was introduced by a resolution step further up in the computation).
As mentioned in the proof sketch in Section 5,
we call two pairs $(\calS,\calRS)$ and $(\calS',\calR_{\calS'})$ {\em equivalent\/}
if one can be obtained from the other by renaming of nulls.
Moreover, let $K$ denote the maximum number of possible values of
non-equivalent pairs $(S,\calRS)$ over the given signature and domain of $D$. By the logspace bound on this data structure, there can only be polynomially many (w.r.t.\ $D$) such values.

Clearly, if all reduced witness trees of $A$
have depth at most $K$, then there can be only finitely many reduced witness trees of $A$ and,
by Lemma \ref{lemma:corresp}, only finitely many (reduced) proof trees. Hence,
$A$ indeed has finite multiplicity. Conversely, suppose that
a ground atom $A$ has a reduced witness tree $\calW$ of
depth greater than $K$. Choose a branch $\pi$ from the root to some leaf node in $\calW$, such that
$\pi$ has maximum length (i.e., its length corresponds to the depth of $\calW$).
By assumption, the depth of $\calW$ (and, therefore, the length of $\pi$) is greater than $K$. Hence, there exist two nodes $N$ and $N'$ on $\pi$ which are labelled with equivalent pairs $(\calS,\calRS)$ and $(\calS',\calR_{\calS'})$. Let
$N$ be an ancestor of $N'$ and let $d$ denote the distance between the two nodes.
Then (after appropriate renaming of nulls) we can
replace the subtree rooted at $N'$ by the subtree rooted at $N$ thus producing a reduced
witness tree whose depth is $\depth(\calW) + d$. By iterating this transformation, we can produce an infinite number of
witness trees of $A$ of depth $\depth(\calW) + d$, $\depth(\calT) +2d$,  $\depth(\calT) +3d$, etc. Hence, $A$ indeed has infinite multiplicity.
This proves the first claim of the theorem.

For the second claim, we proceed in the opposite direction: If
$A$ has infinite multiplicity then, by the first claim,
$A$ has a reduced witness tree of depth greater than $K$. Suppose that
all these reduced witness trees of depth greater than $K$
actually have depth greater than $2K$. Then, we can inspect all branches $\pi$ in the
proof tree and identify nodes $N$ and $N'$ with equivalent labels.
Again, let $N$ be an ancestor of $N'$ and suppose that the distance between the two nodes
is $d$ with $1 \leq d \leq K$. Then (after appropriate renaming of nulls), we can
replace the subtree rooted at $N$ by the subtree rooted at $N'$. In this way,
 at least one branch $\pi$ of
of length $> 2K$ has been replaced by a branch of length $\length(\pi) - d$.
By applying this transformation to
all branches of length $> 2K$, we eventually produce a reduced witness tree
whose depth is in the interval  $[K+1,2K]$.
\end{proof}

\begin{proof}[Proof of Theorem~\ref{theo:finite-multiplicity}, FINITENESS problem]
A decision procedure for the FINITENESS problem can be obtained by modifying the \pt algorithm from
\cite{DBLP:conf/pods/ArenasGP18} in such a way that we ``guess' a branch $\pi$ in the witness tree
with two nodes $N_1$ and $N_2$ on $\pi$ that carry equivalent values
$(\calS_1,\calR_{\calS_1})$ and $(\calS_2,\calR_{\calS_2})$.
To this end, we extend the data structure $(\calS,\calRS)$
propagated in the  \pt algorithm by a pair $(b,y)$, where $b$ is a Boolean flag and $y$ is either
$\varepsilon$ or another copy of the data structure $(\calS,\calR_{\calS})$.
$b = $ true means that we yet have to find
two nodes $N_1$ and $N_2$ (where $N_1$ is the ancestor of $N_2$)
with equivalent labels on a branch in the witness tree.
In this case, either $y = \varepsilon$ (which means that we have not chosen $N_1$ yet) or
$y$ carries the value $(\calS_1,\calR_{\calS_1})$ of the data structure at node $N_1$.
$b = $ false  means that two such nodes $N_1, N_2$ have already been found further up on this branch or are searched for on a different branch.
By appropriately propagating the information in $(b,y)$, we can decide FINITENESS in ALOGSPACE and, hence, in PTIME. Indeed, $(b,y)$ clearly fits into logspace. This additional information
$(b,y)$ is maintained as follows: On the initial call of \pt, we pass $b = $ true (meaning that we
yet have to find the nodes $N_1$ and $N2$). When processing data structure
$(\calS, \calRS)$, we distinguish the following cases:

\begin{itemize}
\item Case 1: $b = $ true and $y = \varepsilon$ (i.e., we have not chosen $N_1$ yet):
we have to (non-deterministically) choose one of the universal branches of \pt with data structure
$(\calS_i,\calR_{\calS_i})$ for some $i$ to which we pass on $b = $ true; to all other universal branches, we pass on $b = $ false. This means that we search for the two nodes $N_1$ and $N_2$ along the
branch the continues with $(\calS_i,\calR_{\calS_i})$. Moreover, for the branch corresponding to
$(\calS_i,\calR_{\calS_i})$, we non-deterministically choose between letting $y = \varepsilon$ (meaning that we will select $N_1$ further below on this branch) and setting
$y = (\calS, \calRS)$ (meaning that we choose $N_1$ with data structure $(\calS, \calRS)$ and we will search for $N_2$ with equivalent value below on this branch).
\item Case 2: $b = $ true and $y =(\calS_1,\calR_{\calS_1})$ (i.e., we have chosen $N_1$ above and we
search for $N_2$ with equivalent value): if $(\calS_1,\calR_{\calS_1})$ and
$(\calS, \calRS)$ are equivalent, then we pass on $b = $ false to all universal branches. This means, that we have indeed found two nodes $N_1,N_2$ in the reduced witness tree with equivalent values.
Otherwise, we have to (non-deterministically) choose one of the universal branches of \pt with data structure $(\calS_i,\calR_{\calS_i})$ for some $i$ to which we pass on $(b,y)$ unchanged.
This means that we have to continue our search for $N_2$ along this branch.
\item Case 3: $b = $ false  (i.e., we have already found $N_1, N_2$ further up on this branch or
we search for $N_1,N_2$ on a different branch in the reduced witness tree). Then we simply pass on
$(b,y)$ unchanged to all universal branches.
\end{itemize}

\noindent
If we eventually reach the base case (i.e., $\calS$ is a singleton consisting of a ground atom
from EDB $D$), then we have to check if $b = $ false. Only in this case, we have an accept node
in the witness tree. Otherwise (i.e., $b = $ true), we reject.
\end{proof}

\begin{proof}[Proof of Theorem~\ref{theo:finite-multiplicity}, MULTIPLICITY problem]
We construct a polynomial-time algo\-rithm for the MULTIPLICITY problem by converting the alternating \pt algorithm from \cite{DBLP:conf/pods/ArenasGP18}
into a deterministic algorithm, were the existential guesses and universal branches are realised by loops over all possible values. At the heart of our algorithm is a procedure $\ptt$, which takes as input a pair $(\calS,\calRS)$ and returns the number $m$ of non-isomorphic, reduced witness trees for this parameter value (referred to as the
``multiplicity'' of $(\calS,\calRS)$ in the sequel).
Moreover, at the end of executing  procedure \ptt, we store the combination of $(\calS,\calRS)$ and  $m$ in a table. Whenever procedure \ptt
is called, we first check if an equivalent value of the input parameter
$(\calS,\calRS)$ (i.e., obtainable by renaming of nulls) already exists in the table. If so, we simply read out the corresponding multiplicity $m$ from the table and return this value. Otherwise, the resolution steps as in the original \pt algorithm are carried out (for all possible combinations of
$\rho_i$ and $\gamma_i$) and we determine the multiplicity $m$
by recursive calls of \ptt.

A high-level description of  PROGRAM {\sf ComputeMultiplicity} with procedure \ptt is given in Figure \ref{fig:ptt}. We assume that this program is only executed if we have checked before that
$P(\bar{c})$ has finite multiplicity.
The program uses two global variables: the table $\Tab$ and a stack $\Stack$. The meaning of the table has already been explained above. The stack $\Stack$ is used to detect if a pair equivalent to the
current input $(\calS,\calRS)$ has already occurred further up in the call hierarchy. If so, we immediately return $m = 0$, because a valid reduced witness tree cannot have such a loop, if we have already verified before that $P(\bar{c})$ has finite multiplicity.
We next check if the multiplicity of $(\calS,\calRS)$ has already been computed
(and stored in $\Tab$) before. If so, we just need to read out the result and return it.
If the base case has been reached (i.e., $\calS$ just contains a ground atom from EDB
$D$), then the return value is simply the multiplicity of this ground atom in $D$.

In all other cases, the multiplicity of $(\calS,\calRS)$
has to be determined via recursive calls of procedure \ptt. To this end, we
first store a copy of $(\calS,\calRS)$ so that we have it available at the end of the procedure, when we want to store $(\calS,\calRS)$ and its multiplicity in $\Tab$.
Moreover, we eliminate from $\calS$ all atoms $A$ which introduce some null $z$,
such that
$\calRS$ contains a pair $(z,A)$. This is required to compute only {\em reduced\/} witness trees. The concrete derivation of atom $A$ is taken care of in another computation path.
We then resolve the remaining atoms $A_1, \dots, A_k$ in all possible ways with rules from $\Pi$. Of course, these resolution steps must be consistent with
$\calRS$. This means that if rule $\rho_i$ introduces a null $z$ and
$\calRS$ contains a pair $(z,x)$ with $x \neq \varepsilon$,
then $x = \gamma_i(\head(\rho_i))$ must hold.
The set of all body atoms resulting from these resolution steps is partitioned
into sets $\{\calS_1, \dots, \calS_\ell\}$ so that all atoms sharing a null
are in the same set $\calS_i$ and the sets $\calS_i$  are minimal with this
property. The set $\calRS$ is updated in the sense that we replace pairs
$(z,\varepsilon)$ by $(z,A)$ if a null $z$ was introduced via atom $A$ by one of the resolution steps.
We recursively call \ptt for all pairs $(\calS_i,\calR_{\calS_i})$, where
$\calR_{\calS_i}$ is the restriction of $\calRS$ to those pairs $(z,x)$, such that
$z$ occurs in $\calS_i$. By construction, the $\calS_i$'s share no nulls. Hence,
the derivations of the $\calS_i$'s can be arbitrarily combined to form a derivation
of $\calS$. We may therefore multiply the number of possible derivations of
the pairs $(\calS_i,\calR_{\calS_i})$ to get the number of derivations of
$(\calS,\calRS)$ for one particular combination
$[(\rho_1,\gamma_1), \dots, (\rho_k,\gamma_k)]$ of resolution steps for the atoms in
$\calS$. The final result of procedure \ptt is the sum of these multiplicities over
all possible combinations
$[(\rho_1,\gamma_1), \dots, (\rho_k,\gamma_k)]$ of resolution steps.
Before returning this final result $m$, we first store
the original input parameter $(\calS',\calR')$ together wtih multiplicity $m$ in
$\Tab$ and remove $(\calS',\calR')$
from the stack.
\end{proof}

\clearpage

\begin{figure}[ht!]
\begin{center}
\fbox{
\begin{minipage}{\textwidth}
\begin{small}

\begin{tabbing}
xx \= xxx \= xxx \= xxx \= xxx \= xxx \= xxx \= xxx \= \kill
{\bf PROGRAM} {\sf ComputeMultiplicity}\\
{\bf Input:} \hskip 7pt  warded \dpm program $\Pi$, EDB $D$ ground atom $P(\bar{c})$. \\
{\bf Output:}  multiplicity of number of $P(\bar{c})$.  \\[1.2ex]
{\bf Global Variables:} Table $\Tab$, Stack $\Stack$. \\[1.5ex]
{\bf Procedure } \ptt \\
// high-level sketch  \\
{\bf Input:} \hskip 7pt  set $\calS$ of atoms, set $\calRS$ of pairs $(z,x)$. \\
{\bf Output:}  number of non-equivalent reduced witness trees with root $(\calS, \calRS)$.  \\[1.2ex]

{\bf begin}  \+ \\
// 1.\ Check for forbidden loop of $(\calS,\calRS)$: \\
{\bf if} $(\calS,\calRS)$ is contained in $\Stack$ {\bf then} {\bf return} 0; \\[1.1ex]
// 2.\ Check if multiplicity of $(\calS,\calRS)$ is already known: \\
{\bf if} equivalent entry of $(\calS,\calRS)$ exists in $\Tab$ with multiplicity $m$
{\bf then} {\bf return} m; \\[1.1ex]
// 3.\ Base case: \\
{\bf if} $\calS =\{A\}$ for a ground atom $A$ from the EDB $D$
{\bf then} {\bf return} multiplicity of $A$ in $D$; \\[1.1ex]
// 4.\ Recursion: \\
push $(\calS,\calRS)$ onto $\Stack$; \\
$(\calS',\calR')$ := $(\calS,\calRS)$; \\
{\bf for each} $A$ such that there exists $(z,A)$ in $\calRS$
 {\bf do} delete $A$ from $\calS$; \\
let $\calS = \{A_1, \dots, A_k\}$;\\
$m$ := 0; \\
{\bf for each} $[(\rho_1,\gamma_1), \dots, (\rho_k,\gamma_k)]$ such that for every $i$,
   $A_i = \gamma_i(\head(\rho_i))$ and \\
   \hskip 33pt $(\rho_i,\gamma_i)$ is consistent with
 $\calRS$  {\bf do} \+  \\
 $\calS^+$ := $\bigcup_{i=1}^k \gamma_i(\body(\rho_i))$; \\
 update $\calRS$ for all nulls that are introduced by one of the rule applications $\rho_i$; \\
 partition $\calS^+$ into sets $\{\calS_1, \dots, \calS_\ell\}$; \\
 compute the corresponding sets of pairs
 $\{\calR_{\calS_1}, \dots, \calR_{\calS_\ell}\}$; \\
{\bf for each} $i$ {\bf do} $m_i$ := \ptt($\calS_i,\calR_{\calS_i}$); \\
$m$ := $m$ + $m_1 \cdot \dots \cdot m_\ell$; \- \\
store $((\calS',\calR'),m)$ in $\Tab$; \\
pop $\Stack$; \\
return $m$; \- \\
{\bf end};  \\[1.5ex]

{\bf begin} (* Main *)  \\
  \> initialize $\Tab$ to empty table; \\
  \> initialize $\Stack$ to empty stack; \\
  \> {\bf output} \ptt ($\{P(\bar{c})\}, \emptyset$); \\
{\bf end}.

\end{tabbing}

\end{small}
\end{minipage}
}
\end{center}

\vspace{-1.0ex}

\caption{Computation of Multiplicity}
\label{fig:ptt}
\end{figure}

\section{Proofs for Section 6}
\label{app:Section6}

\begin{proof}[Proof of Proposition \ref{prop:FOinter}] It remains to prove only the inexpressibility for FOL. First of all, the operation that takes a unary predicate $P$ (with set extension) and creates $P' :=\{\langle a; c\rangle \ |  \ a \in P\}$, where $c$ is a fixed constant, is definable in FOL. So, the elements of $P$ act as local tids for tuples in $P'$, with the latter representing all duplicates for value $c$.

Now, let $P_1, P_2$ be two unary predicates. The majority of elements of $P_1$ belong to $P_2$, by definition, when $|P_1 \cap P_2| \geq |P_1 \smallsetminus P_2|$, which corresponds to the semantics of the {\em majority quantifier} \cite{vb}.  This is equivalent to \ $ (P_1\smallsetminus P_2)' \cap_\nit{mb} (P_1 \cap P_2)' = (P_1 \smallsetminus P_2)'$, where $\cap, \smallsetminus$ are the usual set-theoretic operations.
\ If $\cap_\nit{mb}$ were FOL-definable, so would be the {\em majority quantifier}, which is known to be undefinable in FOL \cite{vb,wester}.\footnote{It is possible to define, the other way around, the minority-based intersection in terms of the majority quantifier.}

For the difference, as we did for $\cap_\nit{mb}$ in (\ref{eq:sqlinter}), we first  define a natural set-operation associated to  $\ominus$, also denoted by $P \ominus Q$, that returns a set containing as many tuples of the form $\langle t; \bar{a}\rangle$ as $|\{\langle t; \bar{a}\rangle~:~ \langle t;\bar{a}\rangle \in P\}| - |\{\langle t; \bar{a}\rangle~:~ \langle t;\bar{a}\rangle \in Q\}|$, and keeps tids  as in $P$.
\ More precisely, it takes as arguments two predicates of the same arity, say $m+1$, whose elements are of the form $\langle t;\bar{a}\rangle$, with $t$ an identifier appearing in position $0$ and only once in the set. In this setting, we assume that there exists a function $f$ that takes a set $\mc{D}(\bar{a}) = \{\langle t_1,\bar{a}\rangle, \ldots,\langle t_n,\bar{a}\rangle\}$  of ``duplicates" of $\bar{a}$, and  a natural number $k$ and returns a subset $\mc{S}(\bar{a})$ of $\mc{D}(\bar{a})$ of size $n-k$ ($0$ if $k \geq n$, i.e. the empty set). Then, $\ominus$ can be defined by:
\begin{equation}
P \ominus Q := \bigcup_{\bar{a}} f(\mc{D}^P(\bar{a}),|\mc{D}^Q(\bar{a})|),
\end{equation}
where $\mc{D}^P(\bar{a}) = \sigma_{1..m = \bar{a}}(P)$, a selection based on the last $m$ attributes, and similarly for  $\mc{D}^Q(\bar{a})$. \ So, $P \ominus Q$ becomes a subset of $P$.

Now, again as for $\cap_{mb}$, for a unary (set) predicate $P$, we define $P' := \{\langle a; c\rangle~:~ a \in P\}$. Now, consider unary predicates $P_1, P_2$ represented as  $P_1', P_2'$, resp. For $P_1, P_2$
we can express the majority quantifier: \ The majority of the elements of $P_1$ are also in $P_2$ \ iff \ $P_1' \smallsetminus (P_1' \ominus P_2')  \subseteq_m  (P_1' \ominus P_2')$. Here, $\smallsetminus$ is the usual set-difference, and $\subseteq_m$ can be defined (on sets) in FOL as:  $A \subseteq_m B :\Leftrightarrow \forall \langle i;e\rangle: \ \langle i;e\rangle \in A \Rightarrow \exists i' \mbox{ with } \langle i';e\rangle \in B$, and for $\langle i_1;e\rangle,\langle i_2;e\rangle \in A$ with $i_1 \neq i_2$, there are $i_1' \neq i_2'$ with  $\langle i_1';e\rangle, \langle i_2';e\rangle  \in B$.
\end{proof}

\ignore{ This result shows, in particular, that a \dpm-based semantics for minority-based multiset intersection cannot be obtained via the uniform  transformation introduced so far of a Datalog-based representation capturing duplicates appearing as the result of this operation.
  \ignore{ Actually, it is also possible to prove, via the non-definability of the ``majority quantifier" in FO predicate logic (FOL) ,  that this intersection is not definable in FO predicate logic (see \cite[sec. 11]{extV} for the details).}
\ignore{The representation in (\ref{eq:sqlinter}) is motivated by the use of arbitrary tids in the general \dpm \ setting, where we are only exploiting the fact that tids are global (guaranteed by the use of fresh and syntactically different nulls whenever
needed). In particular, we do not want to appeal to any ordering or sequential use of nulls. Otherwise,}
 We could redefine (\ref{eq:sqlinter}) by introducing new tids, i.e. tuples $(f(i),\bar{a})$, for each tuple ($i,\bar{a})$ in the condition in (\ref{eq:sqlinter}), with  some function $f$ of tids. The $f(i)$ could be the next tid values after the last one used so far in a list of them. The  operation defined in this would still be non-monotonic.}

 \smallskip

\begin{proof}[Proof of Proposition \ref{prop:stdatinter}] Let $I, P, Q$ be binary predicates whose first arguments hold tids. Let us assume that there is a general Datalog$^{\neg s}$ program $\Pi^I$, with top-predicate $I$, that  defines $\cap_{\nit{mb}}$ over every extensional database (or finite structure)
 of the form $\mc{D} = \langle B, P^\mc{D}, Q^\mc{D}\rangle$, where $B$ is the domain, and
$P^\mc{D}$ and $Q^\mc{D}$ are of the form $\{\langle t_1,a\rangle, \ldots, \langle t_{n_P},a\rangle\}$, $\{\langle t_1',a\rangle, \ldots, \langle t_{n_Q}',a\rangle\}$, respectively.
That is,  $I^\mc{D}$, the extension of $I$ in the standard model of $\Pi^I \cup \mc{D}$ is $P^\mc{D} \cap_{\nit{mb}} Q^\mc{D}$.\ignore{\footnote{If there is such a program, it is not difficult to obtain a Boolean program
instead,  i.e. one with a propositional top-goal, $\nit{true}$, that is computed exactly when $I^\mc{D} = P^\mc{D} \cap_{\nit{mb}} Q^\mc{D}$.}} \
In order to show that such a Datalog$^{\neg s}$-program $\Pi^I$ cannot exists, we assume that it does exist, and we use it to obtain another Datalog$^{\neg s}$-program to define in Datalog$^{\neg s}$-program something else that is provably not definable in Datalog$^{\neg s}$.

Now, consider finite structures of the form $\mc{B} = \langle B, S^\mc{B}\rangle$, with unary $S$. They can be transformed via a Datalog$^{\neg s}$-program into structures of the form
$\mc{D}_B = \langle B, B^c,P_B,Q_B\rangle$, with $B^c = \{\langle b,c\rangle~|~ b \in B\}$, $P_B = \{\langle b,c\rangle~|~ b \in S^\mc{B}\}$, \ $Q_B = B^c \smallsetminus P_B$ (usual set-difference), and $c$ is a fixed, fresh constant. A simple extension of program $\Pi^I$, denoted with $\bar{\Pi}^I$, can be constructed and applied to $\mc{D}_B$ to compute both $P_B \cap_{\nit{mb}} Q_B$ and $Q_B \cap_{\nit{mb}} P_B$ (just extend the original $\Pi^I$ by new rules obtained from $\Pi^I$ itself that exchange the occurrences of predicates $P_B$ and $Q_B$). \ignore{For the latter it is good enough to extend program $\Pi^I$ with simple Datalog rules that invert the position of $P_B$ and $Q_B$ in program $\Pi^I$.}  Then, with one program $\bar{\Pi}^I$, with predicates
   $I_P, I_Q$, we can compute $P\cap_{\nit{mb}} Q$ and $Q\cap_{\nit{mb}} P$, resp. We can add to this program the following three top rules:

\vspace{1mm}\centerline{$\nit{yes} \leftarrow  \nit{yes}_P, \ \nit{yes}_Q. \ \ \ \ \ \ \
\nit{yes}_P \leftarrow I_P(x;a), P(x;a). \ \ \ \ \ \ \
\nit{yes}_Q \leftarrow I_Q(x;a), Q(x;a).$}

\vspace{1mm}
\noindent The second rule obtains $\nit{yes}_P$ when (the extension of) $P\cap_{\nit{mb}} Q$ gives $P$ (because $I_P$ and $P$ have at least an element in common); and the third $\nit{yes}_Q$ when $Q \cap_{\nit{mb}} P$ gives $Q$. The top rule gives $\nit{yes}$ then exactly when the extensions of $P$ and $Q$ have the same cardinality.

All this implies that it is possible to define in Datalog$^{\neg s}$ the class of structures of the form $\mc{B} = \langle B, S^\mc{B}\rangle$ for which $|S^\mc{B}| = |B \smallsetminus S^\mc{B}|$. However, this is not possible
in the logic $L^\omega_{\infty\omega}$ under finite structures \cite[sec. 8.4.2]{ef} (this logic extends FO logic with infinite disjunctions and formulas containing finitely many variables \cite[sec. 3.3]{ef}). Since $L^\omega_{\infty\omega}$ is more expressive than Datalog$^{\neg s}$ (and FO logic and Datalog as a matter of fact) \cite[chap. 8]{ef},
there cannot be a Datalog$^{\neg s}$ program that defines $\cap_{\nit{mb}}$.

The corresponding inexpressibility result for the majority-based difference $\ominus$,
is easily obtained via the equality
$P \cap_{\nit{mb}} Q = P$ iff $P \ominus Q = \emptyset$.
\end{proof}

Notice that proof of Proposition \ref{prop:stdatinter} implies the statement in Proposition \ref{prop:FOinter} for Datalog and FOL. We have kept the two proofs, because they use different techniques from finite model theory.

\section{Further Material for Section 7}

We conclude with a simple example that illustrates the application of our transformation into set semantics of \dpmsn
for SPARQL with bag semantics.

\newcommand{\ans}{\mathit{ans}}

\begin{example}
\label{ex:sparql}

Consider the following SPARQL query from
\url{https://www.w3.org/2009/sparql/wiki/TaskForce:PropertyPaths#Use_Cases}:

\smallskip

\noindent
PREFIX  foaf: <http://xmlns.com/foaf/0.1/> \\
SELECT  ?name \\
WHERE   \mbox{} \hskip 5pt  \{ ?x foaf:mbox <mailto:alice@example> .   \\
\mbox{} \hskip 51pt    ?x foaf:knows/foaf:name ?name  
\mbox{} \hskip 5pt \}

\smallskip

\noindent
This query retrieves the names of people Alice knows.
In our translation into Datalog, we omit the prefix and use the constant {\it alice}
short for <mailto:alice@example> to keep the notation simple.
Moreover, we assume that the RDF data over which the query has to be evaluated
is given by a relational database with a single ternary  predicate $t$ ,
i.e., the database consists of triples $t(\cdot,\cdot,\cdot)$. \
The following Datalog program with answer predicate $\ans$ is equivalent
to the above query:
\begin{equation*}
\ans(N)  \leftarrow  t(X, \mathit{mbox},  \mathit{alice}),
  t(X,  \mathit{knows}, Y), t(Y,  \mathit{name}, N)
\end{equation*}
We now apply our transformation into Datalog$^{\pm}$ from Section~\ref{sec:datalog}. Suppose that
the query has to be evaluated over a multiset $D$ of triples. Recall that
in our transformation, we would first have to extend all triples in $D$ to quadruples
by adding a null (the tid) in the 0-position. Actually, this can be easily automatized
by the first rule in the program below. We thus get the following Datalog$^{\pm}$ program:
\begin{eqnarray*}
\exists Z \ t(Z;U,V,V) & \leftarrow & t(U,V,V).  \\
\exists Z \ \ans(Z;N)  &\leftarrow & t(Z_1; X,  \mathit{mbox},  \mathit{alice}),
t(Z_2; X,  \mathit{knows}, Y), t(Z_3; Y,  \mathit{name}, N).
\end{eqnarray*}

\vspace{-19pt}
\boxtheorem
\end{example}

\end{document}